\documentclass[a4paper]{scrartcl}
\usepackage[utf8]{inputenc}
\usepackage{natbib}
\usepackage{amsthm, amssymb, amsmath}
\usepackage{bm}
\usepackage{dsfont}
\usepackage{enumerate}
\usepackage[usenames,dvipsnames]{color}
\usepackage{caption}
\usepackage{graphicx}
\usepackage[unicode,breaklinks]{hyperref}

\definecolor{linkcolor}{RGB}{15,15,175}
\hypersetup{
    unicode=true,           % non-Latin characters in Acrobat’s bookmarks
    pdftoolbar=true,        % show Acrobat’s toolbar?
    pdfmenubar=true,        % show Acrobat’s menu?
    pdffitwindow=false,     % window fit to page when opened
    pdfstartview={FitH},    % fits the width of the page to the window
    pdfauthor={Christian Sch\"afer},     % author
    pdfcreator={Christian Sch\"afer},   % creator of the document
    pdftitle={On parametric families for sampling binary data with specified mean and correlation},
    pdfkeywords={Binary parametric families}{Sampling correlated binary data}
    pdfnewwindow=true,      % links in new window
    colorlinks=true,        % false: boxed links; true: colored links
    linkcolor=linkcolor,    % color of internal links
    citecolor=linkcolor,    % color of links to bibliography
    filecolor=linkcolor,    % color of file links
    urlcolor =linkcolor     % color of external links
}

\def\B{\mathbb{B}}

\def\R{\mathbb{R}}
\def\Rc{\overline{\R}}

\def\N{\mathbb{N}}

\def\uni{\mathcal{U}} % Uniform
\def\diag{\mathrm{diag}}
\def\tr#1{\mathrm{tr}\left( #1 \right)}
\def\abs#1{\left\vert #1 \right\vert}
\def\det{\mathrm{det}}
\def\set#1{\lbrace #1 \rbrace}

\def\norm#1{\Vert #1 \Vert}
\def\card#1{\vert #1 \vert}
\def\t{^\intercal}

\def\v#1{\bm{#1}}
\def\m#1{\bm{\mathrm{#1}}}
\def\g#1{\m{#1}}

\def\eqdef{:=}

\def\ind{\mathds{1}}
\def\logistic{\ell}

\def\bigtimes{\mathop{\mbox{\Large$\times$}}}

\newcommand{\prob}[1]{\mathbb{P}\left(#1 \right)}
\newcommand{\probplus}[2]{\mathbb{P}_{#1}\left(#2 \right)}

\newcommand{\evplus}[2]{\mathbb{E}_{#1}\left(#2 \right)}

\makeatletter
\newcommand{\raisemath}[1]{\mathpalette{\raisem@th{#1}}}
\newcommand{\raisem@th}[3]{\raisebox{#1}[0pt][0pt]{$#2#3$}}
\makeatother
\def\ExpQu{{\raisemath{-1.1pt}{e}}}
\def\Prod{{\raisemath{-1.1pt}{\sqcap}}}
\def\Link{{\raisemath{-1.1pt}{\,\mu}}}
\def\LogCo{{\raisemath{-1.1pt}{\,\logistic}}}
\def\Gau{{\raisemath{-1.1pt}{gc}}}

\newtheorem{theorem}{Theorem}[section]
\newtheorem{lemma}[theorem]{Lemma}
\newtheorem{proposition}[theorem]{Proposition}
\newtheorem{corollary}[theorem]{Corollary}
\newenvironment{definition}[1][Definition]{\begin{trivlist}
\item[\hskip \labelsep {\bfseries #1}]}{\end{trivlist}}

\def\keywords#1{\par\addvspace\medskipamount{\rightskip=0pt plus1cm
\def\and{\ifhmode\unskip\nobreak\fi\ $\cdot$
}\noindent \textbf{Keywords}\enspace\ignorespaces#1\par}}

\newcommand{\q}[1]{\centering{\colorbox{Yellow}{\parbox{0.9\textwidth}{#1}}}\flushleft{}}
\usepackage{algorithmic}
\usepackage{algorithm}

\floatname{algorithm}{Procedure}

\title{On parametric families for sampling binary data with specified mean and correlation}
\addtokomafont{footnote}{\sffamily\fontsize{6}{3}\selectfont}
\author{Christian Sch\"afer$^{1,2}$}
\def\crest{
\footnotetext[1]{Centre de Recherche en Économie et Statistique, 3 Avenue Pierre Larousse, 92240 Malakoff, France}
\footnotetext[2]{CEntre de REcherches en MAthématiques de la DEcision, Université Paris-Dauphine, Place du Maréchal de Lattre de Tassigny
75775 Paris, France}
}

\def\link{\mu}
\def\Gamma{\varGamma}
\def\q{q}
\def\xi{x}

\begin{document}

\maketitle
\crest
\thispagestyle{empty}

%

% Abstract

%

\begin{abstract}
We discuss a class of binary parametric families with conditional probabilities taking the form of generalized linear models and show that this approach allows to model high-dimensional random binary vectors with arbitrary mean and correlation. We derive the special case of logistic conditionals as an approximation to the Ising-type exponential distribution and provide empirical evidence that this parametric family indeed outperforms competing approaches in terms of feasible correlations.
\end{abstract}
%\subclass{65C05 \and 62E10}

% Introduction

\section{Introduction}
The need to sample random vectors of correlated binary variables arises in various statistical application; examples are estimation of the posterior mean in Bayesian variable selection \citep{george_mcculloch_97}, small-sample properties of estimators in longitudinal studies \citep[for a recent review]{farrell2008methods}, stochastic binary optimization in combinatorics \citep{Rub:CE2}, simulation of ferromagnetic materials \citep{swendsen1987nonuniversal}, performance of neural networks \citep{lebbah2008probabilistic} and market segmentation analysis \citep{dolnicar2001behavioral} among others.

Let $\B\eqdef\set{0,1}$ denote the binary space. In some cases, such as small-sample analysis in longitudinal studies, we need a parametric family $\q$ explicitly for sampling data on $\B^{d}$ with specified mean and correlations. In other cases, the parametric family serves as a proxy for a more complex distribution we cannot directly sample from. Suppose we have two functions $\tilde\pi\colon\B^d\to\R_+$ and $f\colon\B^d\to\R$ and we want to compute the expected value $\evplus{\pi}{f(\v \Gamma)}=h^{-1}\,\sum_{\v\gamma\in\B^{d}}f(\v\gamma)\tilde\pi(\v\gamma)$ with $h\eqdef\sum_{\v\gamma\in\B^{d}}\tilde\pi(\v\gamma)$.

If $d$ is too large for enumeration of the state space we have to rely on Monte Carlo algorithms, the vast majority of which involve sampling Markov transitions with invariant measure $\pi$, the standard approach being the Metropolis-Hastings kernel \citep[ch. 7]{RobCas}. For a transition from $\v X\sim\pi\eqdef\tilde\pi/h$, we sample $\v\Gamma\sim \q(\cdot\mid\v X)$ from an auxiliary kernel $\q$ and accept the step to $\v \Gamma$ with probability
\begin{equation}
\label{eq:acc prob}
\lambda_{\q}(\v\Gamma,\v X)\eqdef \min\set{1,[\tilde\pi(\v\Gamma)\q(\v X\mid \v\Gamma)]/[\tilde\pi(\v X)\q(\v\Gamma\mid \v X)]} 
\end{equation}
or return $\v X$ otherwise. Random walks on $\B^{d}$ are easy to implement but often suffer from slow mixing; independent proposals $\v\Gamma\sim \q$ provide fast-mixing if $\lambda_{\q}(\v\Gamma,\v X)$ is reasonably high on average, in other words if $\q$ is sufficiently close to $\pi$ \citep{schaefer2011sequential,schaefer2012particle}. This rationale complements other approaches to fast mixing such as parallel chains \citep[among others]{bottolo2010ess} or self-avoiding dynamics \citep{hamze2011selfavoiding}.\\

%\subsection{Scope}
The vast field of potential applications in Monte Carlo algorithms encourages the study of families with $d(d+1)/2$ parameters which, like the multivariate normal, accommodate all valid combinations of means and correlations. This paper elaborates some theoretical background on random binary vectors, proves the range of possible correlations for a particular class of parametric families, connects to existing work in the literature and provides broad numerical insight concerning the range of dependencies achievable in practice. It is structured as follows.

In Section \ref{sec:properties}, we introduce suitable notation and review results relating binary distributions to its moments. Section \ref{sec:families from glms} elaborates on parametric families which have, by definition, conditional distributions that are generalized linear regressions. We show that they accommodate the whole range of possible correlations. Section \ref{sec:log cond family} motivates the use of the logistic link function as an approximation to the Ising-type exponential quadratic family. In Section \ref{sec:sampling}, we discuss how to adjust the parametric families to specified marginals. Finally, in Section \ref{sec:numerical experiments} we perform numerical experiments to compare competing approaches for sampling correlated binary data in high dimensions.

%

% Preliminaries on binary data

%

\section{Preliminaries on random binary vectors}
\label{sec:properties}
We write $\B\eqdef\set{0,1}$ for the binary space and denote by $d\in\N$ the generic dimension. Given a vector $\v\gamma\in\B^d$ and an index set $I\subseteq D\eqdef \set{1,\dots,d}$, we write $\v\gamma_I\in\B^{\card I}$ for the sub-vector indexed by $I$ and $\v\gamma_{-I}\in\B^{d-\card I}$ for its complement. For $I=\set{i,\dots,j}$ we use the more explicit notation $\v\gamma_{i:j}$. Unless otherwise defined, $\pi$ denotes an arbitrary probability mass function on $\B^{d}$. We denote by $\evplus{\pi}{f(\v \Gamma)}$ the expected value with respect to $\v \Gamma\sim\pi$ and write $\probplus{\pi}{A}\eqdef\evplus{\pi}{\ind_{A}(\v \Gamma)}$ for an event $A\subseteq \B^{d}$.
\begin{definition}
Let $\v m\in(0,1)^{d}$ be a mean vector. We call
$
\q_{\v m}^{\Prod}(\v \gamma)\eqdef\prod_{i\in D}m_i^{\gamma_i}(1-m_i)^{1-\gamma_i}
$
the product family or the mass function of $d$ independent Bernoulli variables.
\end{definition}

\subsection{Absolute cross-moments}
\begin{definition}
For a set $I\subseteq D$, we refer to
$
m_I^{\pi}\eqdef\evplus{\pi}{\textstyle\prod_{i\in I}\Gamma_i}=\textstyle\sum_{\v\gamma\in\B^{d}}\pi(\v\gamma)\prod_{i\in I}\gamma_i
$
as the cross-moment indexed by $I$.
\end{definition}
Note that $m^{\pi}_I=\probplus{\pi}{\v \Gamma_{I}=\v 1}$ which means that cross-moments and marginal probabilities indexed by $I\subseteq D$ are identical. Higher order cross-moments coincide with first order cross-moments. The range of possible cross-moments is limited by the following constraints.
\begin{proposition}
\label{prop:bin bounds}
The cross-moments  of binary data fulfill the sharp inequalities
\begin{equation}
\label{eq:bin bounds}
\max\left\lbrace\textstyle\sum_{i\in I}m_i-\vert I\vert+1, 0\right\rbrace
\leq m_I
\leq \min\set{m_K\colon K\subseteq I}.
\end{equation}
\end{proposition}
\begin{proof}
The lower bound follows from
\begin{align*}
\textstyle
\abs I -1
=\sum_{\v\gamma\in\m\B^d}(\vert I\vert-1)\pi(\v\gamma) 
\geq\sum_{\v\gamma\in\m\B^d}\left(\sum_{i\in I}\gamma_i-\prod_{i\in I}\gamma_i\right)\pi(\v\gamma)
=\sum_{i\in I}m_i-m_I,
\end{align*}
the upper bound is the monotonicity of the measure.
\end{proof}
For the special case $\abs{I}=2$, Proposition \ref{eq:bin bounds} is a well-known result and has been invoked in several articles dealing with correlated binary data. For the general case, we remark that a mapping $f\colon[0,1]^{\abs{I}}\to[0,1],\ f_{I}(m_{i_1},\dots,m_{i_{\abs{I}}})=m_{I}$, which assigns a cross-moment $m_{I}$ for $I\subseteq D$ as function of the marginals $m_i$ for $i\in I$, is quite similar to a $\abs{I}$-dimensional copula and the inequalities \eqref{eq:bin bounds} are exactly the Fr\'echet-Hoeffding bounds \citep[ch. 2]{nelsen2006introduction}.
\begin{definition}
\label{def:cross-moment matrix}
We say a $d\times d$ symmetric matrix $\m M\eqdef(m_{ij})$ with entries in $(0,1)$ is a \emph{cross-moment matrix of binary data} if $\m M - \diag(\m M)\diag(\m M)\t$ is positive definite and condition \eqref{eq:bin bounds} holds for all $I\subseteq D$ with $\abs{I}=2$.
\end{definition}
We derive the family of distributions which, under the constraints that $\pi$ has given cross-moments, maximizes the entropy $H(\pi)=-\sum_{\v\gamma\in\B^{d}}\pi(\v\gamma)\log[\pi(\v\gamma)]$. The following proposition is just a special case of a more general concept \citep{soofi1994capturing}.
\begin{proposition}
\label{prop:maxent}
Let $\mathcal I\subseteq 2^{D}$ be a family of index sets such that $\set{m_I\colon I\in\mathcal I}$ is a valid set of cross-moments. The maximum entropy distribution having the specified cross-moments has the form $\q(\v \gamma)=\exp(\sum_{I\in \mathcal I}a_I\prod_{i\in I}\gamma_i)/ [\sum_{\v\gamma\in\B^{d}}\exp(\sum_{I\in \mathcal I}a_I\prod_{i\in I}\gamma_i)]$.
\end{proposition}
\begin{proof} Define the Lagrange multipliers
$L(\pi,\v a)=\sum_{I\in \mathcal I}a_{I}[\sum_{\v\gamma\in\B^{d}}\pi(\v\gamma)\prod_{i\in I}\gamma_i-m_I]$
and differentiate $\partial[H(\pi)+L(\pi,\v a)]/\partial \pi(\v\gamma)=-\log[\pi(\v\gamma)]-1+\sum_{I\in \mathcal I}a_{I}\prod_{i\in I}\gamma_i$. Solving the first order condition and normalizing completes the proof.
\end{proof}

\subsection{Standardized cross-moments}
\begin{definition}
For a set $I\subseteq D$, we define
$
\textstyle u_I^{\pi}(\v\gamma)\eqdef \prod_{i\in I}(\gamma_i-m_i^{\pi})[m_i^{\pi}(1-m_i^{\pi})]^{-1/2}
$
and refer to $c_I^{\pi}\eqdef\evplus{\pi}{u_I^{\pi}(\v \Gamma)}$ as the (generalized) correlation coefficient indexed by $I$.
\end{definition}

A $d\times d$ positive definite matrix $\m C$ with entries in $[-1,1]$ and $\diag(\m C)=\v 1$ is not the correlation matrix of a binary distribution for every mean vector $\v m\in(0,1)^{d}$. In fact, $\m C$ is a correlation matrix if and only if $\m M=\m C\cdot\v s\v s\t+\v m\v m\t$ is valid in the sense of Definition \ref{def:cross-moment matrix}, where the dot means point-wise multiplication and $s_i^{2}\eqdef m_{i}(1-m_{i})$. \cite{chaganty2006range} elaborate alternative conditions for compatibility between correlations and means, but these do not seem easier to express or to check.

In the context of binary data, the notion of ``strong correlations'' refers to correlation coefficients which are at the boundary of the feasible range with respect to the mean vector. Note that the absolute value of the correlation coefficient does, in itself, not tell whether the correlation is easy or difficult to model. The following statement relates the notions of uncorrelated and independent variables.

\begin{proposition}
Let $\v X$ be a $d$-dimensional binary random vector. For $d=2$, entries are uncorrelated if and only if they are independent. For $d\geq 3$, entries might be mutually uncorrelated but not independent.
\end{proposition}
\begin{proof}
Let $p_{x_1x_2}\eqdef\prob{\Gamma_1=x_1,\Gamma_2=x_2}$. By definition $p_{11}=m_{12}=m_1m_2$. Further, we obtain $p_{10}=m_1-m_{12}=m_1(1-m_2)$ and, analogously, $p_{01}=(1-m_1)m_2$. Finally, we have $p_{00}=1+m_{12}-m_{1}-m_{2}=(1-m_1)(1-m_2)$. For $d\geq3$, let for instance $p_{000}=p_{011}=p_{101}=p_{110}=1/4$ and $p_{100}=p_{010}=p_{001}=p_{111}=0$. The entries are mutually uncorrelated, but not independent since $p_{111}=0\neq 1/8=m_1m_2m_3$.
\end{proof}

The following representation by \citet{bahadur61representation} allows to write a binary distribution in terms of its generalized correlation coefficients.
\begin{proposition}
\label{prop:bahadur}
Let $\pi$ be a binary distribution with mean $\v m\in(0,1)^{d}$. Then,
\begin{equation*}
\pi(\v\gamma)=\q^\Prod_{\v m}(\v\gamma)\, \left[\textstyle\sum_{I\subseteq D} c_I^{\pi} u_I^{\pi}(\v\gamma)\right].
\end{equation*}
\end{proposition}
\begin{proof}
We give the proof by \citet{bahadur61representation} using the notation introduced above. The set $\set{u_I^{\pi}\colon I\subseteq D}$ forms an orthonormal basis on $\mathcal F\eqdef\set{f\colon \B^{d}\to\R}$ with respect to the inner product
$
\textstyle
(f,g)=\evplus{\q^\Prod_{\v m}}{f(\v \Gamma)g(\v X)}=\sum_{\v\gamma\in \B^d} f(\v\gamma)g(\v\gamma)\q^\Prod_{\v m}(\v\gamma).
$
Therefore, every function $f\in\mathcal F$ has a unique representation $f(\v\gamma)=\sum_{I\subseteq D}(f,u_{I}^{\pi})u_{I}^{\pi}(\v\gamma)$. Compute the inner products
$
\textstyle
(\pi/\q^\Prod_{\v m},u_I^{\pi})
=\sum_{\v\gamma\in \B^d} [\pi(\v\gamma)/\q^\Prod_{\v m}(\v\gamma)]u_I^{\pi}(\v\gamma)\q^\Prod_{\v m}(\v\gamma)
=\evplus{\pi}{u_I^{\pi}(\v \Gamma)}=c_I^{\pi}
$
to obtain the desired form $\pi(\v\gamma)/\q^\Prod_{\v m}(\v\gamma)=\sum_{I\subseteq D} c_I^{\pi}u_I^{\pi}(\v\gamma)$.
\end{proof}
Using Proposition \ref{prop:bahadur}, we may bound the $l^{p}$ distance between two binary distribution with the same mean in terms of nearness of their correlation coefficients.
\begin{proposition}
Let $\pi$ and $\omega$ be binary distributions with mean $\v m\in(0,1)^{d}$. For $p\geq1$,
\begin{equation*}
\textstyle
\sum_{\v\gamma\in\B^{d}}\abs{\pi(\v\gamma)-\omega(\v\gamma)}^{p}
\leq \sum_{I\subseteq D}2^{(1-\min\set{p,2})\abs{I}} \vert c_{I}^{\pi}-c_{I}^{\omega}\vert^{p}
\leq (1+r)^{d}-dr-1
\end{equation*}
where $r=2^{1-\min\set{p,2}}\max_{I \subseteq D} \vert c_{I}^{\pi}-c_{I}^{\omega}\vert^{p/\abs{I}}$.
\end{proposition}
\begin{proof}
Since $u_{I}^{\pi}=u_{I}^{\omega}$ for all $I\subseteq D$, applying Proposition \ref{prop:bahadur} yields
\begin{align*}
\textstyle
\sum_{\v\gamma\in\B^{d}}\abs{\pi(\v\gamma)-\omega(\v\gamma)}^{p}
&=    \textstyle\sum_{\v\gamma\in\B^{d}}\abs{\q^\Prod_{\v m}(\v\gamma)\sum_{I\subseteq D}u_I^{\pi}(\v\gamma)(c_I^{\pi}-c_I^{\omega})}^{p} \\
&\leq \textstyle\sum_{I\subseteq D} \vert c_I^{\pi}-c_I^{\omega} \vert^{p}\, \evplus{\q^{\Prod}_{\v m}}{\abs{u_I^{\pi}(\v \Gamma)}^{p}}.
\end{align*}
Using that $x^{p-1}+(1-x)^{p-1}\leq 2^{2-\min\set{p,2}}$ for all $x\in(0,1)$, we obtain the bound
\begin{equation*}
\textstyle\evplus{\q^{\Prod}_{\v m}}{\abs{u_I^{\pi}(\v \Gamma)}^{p}}
\leq \prod_{i\in I} [m_i(1-m_i)]^{1/2}[m_i^{p-1}+(1-m_i)^{p-1}]
\leq 2^{(1-\min\set{p,2})\abs{I}}.
\end{equation*}
Finally, we have $\sum_{I\subseteq D}2^{(1-\min\set{p,2})\abs{I}} \vert c_{I}^{\pi}-c_{I}^{\omega}\vert^{p}\leq \sum_{I\subseteq D, \abs{I}\geq2} r^{\abs{I}}=(1+r)^{d}-dr-1$, since by definition $c_{I}^{\pi}=c_{I}^{\omega}$ for all $I\subseteq D$ with $\abs{I}\leq2$.
\end{proof}
% \begin{corollary}
% Let $\pi$ and $\q$ be binary distributions with mean $\v m\in(0,1)^{d}$. The total variation distance between $\pi$ and $\q$ is bounded by $\frac{1}{2}\sum_{I\subseteq D} \vert c_{I}^{\pi}-c_{I}^{\q}\vert$.
% \end{corollary}
\begin{corollary}
\label{corr: cross moments}
Let $\pi$ and $\q$ be binary distributions with cross-moment matrix $\m M$. Then we have 
$\sum_{\v\gamma\in\B^{d}}\abs{\pi(\v\gamma)-\q(\v\gamma)}^{p}\leq (1+r)^{d}-\frac{1}{2}d(d-1)r^{2}-dr-1$.
\end{corollary}
With regard to the Metropolis-Hastings kernel mentioned in the introductory section, the factor $\frac{1}{2}d(d-1)r^{2}$ in Corollary \ref{corr: cross moments} is the gain of a more complex proposal distribution $\q_{\m M}$ with $\m M=\m M^{\pi}=\m M^{\q}$ over a simple product model $\q_{\v m}^{\Prod}$ with $\v m=\v m^{\pi}=\v m^{\q}$.

The following result shows how the cross-moments of the proposal distribution affect the auto-covariance of the independent Metropolis-Hastings sampler.
\begin{proposition}
Let $\pi$ and $\q$ be binary distributions with mean $\v m\in(0,1)^{d}$ and denote by $\kappa(\v\gamma\mid\v x)\eqdef \q(\v\gamma)\lambda_{\q}(\v\gamma,\v x)+\delta_{\v x}(\v \gamma)[1-\textstyle\sum_{\v y\in\B^{d}}\q(\v y)\lambda_{\q}(\v y,\v x)]$ the Metropolis-Hastings kernel with invariant measure $\pi$ and proposal distribution $\q$ where $\lambda_{\q}(\cdot,\v x)$ is defined in \eqref{eq:acc prob}. The auto-covariance between $\v X\sim\pi$ and $\v \Gamma\sim\kappa(\cdot\mid \v X)$ is
\begin{equation*}
\evplus{\kappa,\pi}{\v \Gamma\v X\t}-\v m \v m\t=\frac{1}{2}(\m M^{\pi}-\m M^{\q})+\m R^{\kappa}
\end{equation*}
with $\m R^{\kappa}=(r^{\kappa}_{ij})$ where $\vert r^{\kappa}_{ij}\vert\leq \sum_{\v \gamma\in\B^{d}}\abs{\pi(\v\gamma)-\q(\v\gamma)}$.
\end{proposition}

\begin{proof}
We plug the definition of the kernel into the expected value and obtain
\begin{align*}
\evplus{\kappa,\pi}{\v \Gamma\v X\t}
&=\sum_{\v \gamma,\,\v x\in\B^{d}} \gamma_ix_j \kappa(\v\gamma\mid \v x)\pi(\v x) \\
&=\sum_{\v \gamma,\,\v x\in\B^{d}} \gamma_ix_j \q(\v\gamma)\lambda_{\q}(\v\gamma,\v x)\pi(\v x)
+\sum_{\v x\in\B^{d}}x_ix_j[1-\textstyle\sum_{\v y\in\B^{d}}\q(\v y)\lambda_{\q}(\v y,\v x)]\pi(\v x) \\
&=m_{ij}^{\pi}+\sum_{\v \gamma,\,\v x\in\B^{d}} (\gamma_ix_j-x_ix_j) \q(\v\gamma)\pi(\v x)\lambda_{\q}(\v\gamma,\v x) \\
&=m_im_j+\frac{1}{2}(m_{ij}^{\pi}-m_{ij}^{\q})+\frac{1}{2}\sum_{\v \gamma,\,\v x\in\B^{d}} (\gamma_ix_j-x_ix_j) \abs{\q(\v\gamma)\pi(\v x)-\q(\v x)\pi(\v\gamma)},
\end{align*}
where we used $2\q(\v\gamma)\pi(\v x)\lambda_{\q}(\v\gamma,\v x)=\q(\v\gamma)\pi(\v x)+\q(\v x)\pi(\v\gamma)-\abs{\q(\v\gamma)\pi(\v x)-\q(\v x)\pi(\v\gamma)}$.
The triangle inequality
\begin{align*}
&\sum_{\v \gamma,\,\v x\in\B^{d}} \abs{\q(\v\gamma)\pi(\v x)-\q(\v x)\pi(\v\gamma)}
=\!\!\sum_{\v \gamma,\,\v x\in\B^{d}} \abs{\q(\v\gamma)\pi(\v x)-\pi(\v\gamma)\pi(\v x)+\pi(\v\gamma)\pi(\v x)-\q(\v x)\pi(\v\gamma)} \\
\leq
&\sum_{\v \gamma,\,\v x\in\B^{d}} \left[\abs{\q(\v\gamma)-\pi(\v\gamma)}\pi(\v x)+\abs{\pi(\v x)-\q(\v x)}\pi(\v\gamma)\right]
=2\sum_{\v \gamma\in\B^{d}}\abs{\pi(\v\gamma)-\q(\v\gamma)}.
\end{align*}
yields the bound on $r^{\kappa}_{ij}\eqdef\frac{1}{2}\sum_{\v \gamma,\,\v x\in\B^{d}} (\gamma_ix_j-x_ix_j) \abs{\q(\v\gamma)\pi(\v x)-\q(\v x)\pi(\v\gamma)}$.
\end{proof}
% For a proposal distribution $\q_{\m M}$ with $\m M=\m M^{\pi}=\m M^{\q}$, the auto-covariance first term vanishes and the remainders $\vert r^{\kappa}_{ij}\vert$ are, on average, smaller as implied by Corollary \ref{corr: cross moments}.

\subsection{Structured correlations}

For some applications, it suffices to model structured dependencies, such as exchangeable $(c_{ij}=c)$, moving average $(c_{ij}=c\ind_{\abs{i-j}=1})$ or autoregressive $(c_{ij}=c^{\abs{i-j}})$ correlations for $i\neq j\in D$. There is a long series of articles concerned with efficient approaches to sampling binary vectors for structured correlations \citep{farrell2006nonlinear,qaqish2003family,oman2001modelling,lunn1998note,park1996simple}. In this paper, we focus on the problem of sampling binary data with arbitrary cross-moment matrix.

\section{Parametric families based on generalized linear models}
\label{sec:families from glms}
We want to construct a parametric family $\q$ for sampling independent random vectors with specified mean and correlations. Sampling in high dimensions, however, requires the computation of conditional distributions $\q(\gamma_i\mid \v\gamma_{1:i-1})$, and it is therefore convenient to define the parametric family directly in terms of its conditionals.
\begin{definition}
Let $\link\colon\overline\R\to[0,1]$ be a monotonic function and $\m A\eqdef(a_{ij})$ a $d\times d$ real-valued lower triangular matrix. We refer to
\begin{align*}
\textstyle
\q_{\m A}^{\Link}(\v\gamma)=\prod_{i=1}^d \left[\link(a_{ii}+\sum_{j=1}^{i-1}a_{ij}\gamma_j)\right]^{\gamma_i}\left[1-\link(a_{ii}+\sum_{j=1}^{i-1}a_{ij}\gamma_j)\right]^{1-\gamma_i},
\end{align*}
as the $\link$-conditionals family.
\end{definition}
\begin{proposition}
\label{cor:mu product}
Let $\link\colon\overline\R\to[0,1]$ be a monotonic bijection and $\v m\in(0,1)^d$ a mean vector. For $\m A=\diag[\link^{-1}(\v m)]$ we have $\q_{\m A}^{\Link}=\q_{\v m}^{\Prod}$.
\end{proposition}
By construction, it is straightforward to sample $\v x\sim \q_{\m A}^{\Link}$ and evaluate $\q_{\m A}^{\Link}(\v x)$ point-wise as summarized in Procedure \ref{algo:sampling}. Alternatively, one could sample from an auxiliary distribution $\varphi$ on $\R^d$ which allows to compute $\varphi(x_i\mid \v x_{1:i-1})$ and define a parametric family $\q_{\tau,\varphi}(\v\gamma)=\int_{\tau^{-1}(\v \gamma)} \varphi(\v x) d\v x$ through the mapping $\tau\colon\R^d\to\B^d$. We come back to this idea in Section \ref{sec:Gaussian}.
\floatname{algorithm}{Procedure}
\begin{algorithm}[ht]
\caption{Sampling from a $\mu$-conditionals family}
\begin{algorithmic}
\STATE $\v x=(0,\dots,0),\ p\gets 1$
\FOR {$i=1\dots,d$}\vspace{0.2em}
  \STATE $c\gets \q_{\m A}^{\Link}(x_i=1\mid\v x_{1:i-1})=\link({a_{ii}+\sum_{j=1}^{i-1}a_{ij}x_j})$, $u\gets U\sim\mathcal{U}_{[0,1]}$ \\[.2em]
  \STATE \bf{ if } $u<c$ \bf{ then }$x_i\gets1$ \\
  \STATE $p\gets\begin{cases}
	    p\cdot c     & \textbf{if }\ \ x_i=1 \\
	    p\cdot (1-c) & \textbf{if }\ \ x_i=0
            \end{cases}$ \\
\ENDFOR
\RETURN $\v x,\ p$
\end{algorithmic}
\label{algo:sampling}
\end{algorithm}

\cite{qaqish2003family} discusses the $\mu$-conditionals family with a truncated linear link function $\mu(x)=\min\set{\max\set{x,0},1}$. The linear structure allows to compute the parameters by simple matrix inversion; on the downside, the linear function is truncated and fails to accommodate complicated correlation structures. Therefore, \cite{qaqish2003family} elaborates on conditions that guarantee the linear conditionals family to be valid for special correlation structures.

\cite{farrell2006nonlinear} propose a $\mu$-conditionals family with a logistic link function $\mu(x)=1/[1+\exp(-x)]$. However, they only analyze the special case of autoregressive correlation structure. In Section \ref{sec:log cond family}, we further motivate the use of the logistic link function which indeed allows to model any feasible correlation structure as states the following theorem.

\begin{theorem}
\label{thm:mean+corr}
Let $\link\colon\Rc\to[0,1]$ be a monotonic, differentiable bijection and $\m M$ a $d\times d$ cross-moment matrix. There is a unique $d\times d$ real-valued lower triangular matrix $\m A$ such that
$\sum_{\v\gamma\in\B^d} \q_{\m A}^{\Link}(\v\gamma)\v\gamma\v\gamma\t=\m M$.
\end{theorem}

Besides the logistic function invoked above, popular link functions include the complementrary log-log function with $\mu(x)=1-\exp[-\exp(x)]$ and the probit function with $\mu(x)=(2\pi)^{-1/2}\int_{-\infty}^{x}\exp(-y^{2}/2)dy$ \citep[sec. 4.3]{mccullagh1989generalized}. We derive two auxiliary results to structure the proof of Theorem \ref{thm:mean+corr}.

\begin{lemma}
\label{lem:ext cross-moments}
For a cross-moment matrix $\m M$ with mean vector $\v m=\diag(\m M)$, we have
\begin{equation*}
\begin{pmatrix}
\m M & \v m \\
\v m\t & 1
\end{pmatrix}>0.
\end{equation*}
\end{lemma}
\begin{proof}
Note that $\v m\t\m M^{-1}\v m - (\v m\t\m M^{-1}\v m)^2=(\m M^{-1}\v m)\t(\m M - \v m\v m\t)\m M^{-1}\v m>0$ because the covariance matrix $\m M-\v m\v m\t$ is positive definite. Dividing by $\v m\t\m M^{-1}\v m>0$ we obtain $1-\v m\t\m M^{-1}\v m > 0$ which yields
\begin{align*}
\det{
\begin{pmatrix}
\m M & \v m \\
\v m\t & 1
\end{pmatrix}}
&=
\det\left[
\begin{pmatrix}
\m M & \v 0 \\
\v 0\t & 1
\end{pmatrix}
\begin{pmatrix}
\m I & \m M^{-1}\v m \\
\v m\t & 1
\end{pmatrix}\right]
=
\det(\m M)
\det
\begin{pmatrix}
\m I & \m M^{-1}\v m \\
\v 0\t & (1-\v m\t\m M^{-1}\v m)
\end{pmatrix} \\[1ex]
&=\det(\m M)(1-\v m\t\m M^{-1}\v m)>0.
\end{align*}
Therefore, all principal minors are positive.
\end{proof}

\begin{lemma}
\label{lem:bijective}
Let $\link\colon\Rc\to[0,1]$ be a monotonic, differentiable bijection, and denote by $B_r^{n}=\set{\v x\in\R^{n}\mid \v x\t \v x < r^2}$ the open ball with radius $r>0$. Let $\pi$ be a binary distribution with cross-moment matrix $\m M$. We write $\v m=\diag(\m M)$ and $\v m^{*}=(\v m\t,1)\t$ for the mean vector. There is $\varepsilon_r>0$ such that the function
\begin{equation*}
f\colon B_r^{d+1}\to\bigtimes_{i=1}^{d+1}(\varepsilon_r,m_i^*-\varepsilon_r),\quad
f(\v a)=\sum_{\v\gamma\in\B^d}\pi(\v\gamma)\link(a_{d+1}+\textstyle\sum_{k=1}^{d}a_k\gamma_k)
\begin{pmatrix}
\v\gamma \\ 1
\end{pmatrix}
\end{equation*}
is a differentiable bijection.
\end{lemma}
\begin{proof}
We set
$\varepsilon_r\eqdef\max\bigcup_{i\in D\cup\set{d+1}}\left\{\min_{\v a\in B_r^{d+1}} f_i(\v a),\ m_i^*-\max_{\v a\in B_r^{d+1}} f_i(\v a)\right\}.$
For $i,j\in D\cup\set{d+1}$, the partial derivatives of $f$ are
\begin{align*}
\frac{\partial f_i}{\partial a_j}=\sum_{\v\gamma\in\B^d}\pi(\v\gamma)\link'(a_{d+1}+\textstyle\sum_{k=1}^{d}a_k\gamma_k)
\times\begin{cases}
\gamma_i\gamma_j & (i,j\in\set{1,\dots,d}) \\
\gamma_i & (j=d+1) \\
\gamma_j & (i=d+1) \\
1        & (i=j=d+1).
\end{cases}
\end{align*}
We have $\eta_r\eqdef\min_{\v a\in B_r^{d+1}}\min_{\v\gamma\in\B^d}
\link'(a_{d+1}+\textstyle\sum_{i=1}^{d}a_i\gamma_i)>0$ since $\link$ is strictly monotonic.
Then the Jacobian is positive for all $\v a\in B_r^d$,
\begin{align*}
\det{f'(\v a)}
=\det\left[\sum_{\v\gamma\in\B^d}\pi(\v\gamma)
\link'(a_{d+1}+\textstyle\sum_{i=1}^{d}a_i\gamma_i)
\begin{pmatrix}
\v\gamma\v\gamma\t & \v\gamma \\
\v\gamma\t & 1
\end{pmatrix}\right]
\geq
\eta_r^{d+1}\,
\det
\begin{pmatrix}
\m M & \v m \\
\v m\t & 1
\end{pmatrix}
>0,
\end{align*}
where we applied Lemma \ref{lem:ext cross-moments} in the last inequality.
\end{proof}

\begin{proof}[Theorem \ref{thm:mean+corr}]
We proceed by induction over $d$. For $d=1$, $\m A(1)$ is a scalar and we define the $\mu$-conditionals family $\q_{\m A(1)}^{\Link}$ via Corollary \ref{cor:mu product}. Suppose that we have already constructed a $\mu$-conditionals family $\q_{\m A(d)}^{\Link}$ with $d\times d$ lower triangular matrix $\m A(d)$ and cross-moment matrix $\m M(d)$. We can add a new dimension to the $\mu$-conditionals model $\q_{\m A(d)}^{\Link}$ without changing $\m M(d)$, since
\begin{align*}
\sum_{\v\xi\in\B^{d+1}} \q_{\m A(d+1)}^{\Link}(\v\xi)\v\xi\v\xi\t
=&\ \sum_{\v\xi\in\B^{d+1}} \q_{\m A(d)}^{\Link}(\v\xi_{1:d})\v\xi\v\xi\t
\left[\link(a_{d+1,d+1}+\textstyle\sum_{j=1}^{d} a_{d+1,j}\xi_j)\right]^{\xi_{d+1}}\times\\ &\qquad
\left[1-\link(a_{d+1,d+1}+\textstyle\sum_{j=1}^{d} a_{d+1,j}\xi_j)\right]^{1-\xi_{d+1}} \\
=&\ \sum_{\v\gamma\in\B^d} \q_{\m A(d)}^{\Link}(\v\gamma)
\Bigg\{
\link(a_{d+1,d+1}+\textstyle\sum_{j=1}^{d} a_{d+1,j}\gamma_j)
\begin{pmatrix}
\v\gamma\v\gamma\t & \v\gamma \\
\v\gamma\t  & 1
\end{pmatrix}+\\ &\qquad
\left[1-\link(a_{d+1,d+1}+\textstyle\sum_{j=1}^{d} a_{d+1,j}\gamma_j)\right]
\begin{pmatrix}
\v\gamma\v\gamma\t & \v0 \\
\v0\t  & 0
\end{pmatrix}
\Bigg\} \\
=&\
\sum_{\v\gamma\in\B^d} \q^{\Link}_{\m A(d)}(\v\gamma)
\link(a_{d+1,d+1}+\textstyle\sum_{j=1}^{d} a_{d+1,j}\gamma_j)
\begin{pmatrix}
\m 0   & \v\gamma \\
\v\gamma\t  & 1
\end{pmatrix}+\\ &\qquad
\begin{pmatrix}
\m M(d) & \v0 \\
\v0 \t  & 0
\end{pmatrix}
\end{align*}
For reasons of symmetry, it suffices to show that there is $\v a\in \R^{d+1}$ such that
\begin{equation*}
f(\v a)=\sum_{\v\gamma\in\B^{d}} \q_{\m A(d)}(\v\gamma)\link(a_{d+1}+\textstyle\sum_{i=1}^{d}a_i\gamma_i)
\begin{pmatrix}
\v\gamma \\ 1
\end{pmatrix}=\m M(d+1)_{\bullet d+1},
\end{equation*}
where the r.h.s. denotes the $(d+1)$th column of the augmented cross-moment matrix. There is $\varepsilon>0$ so that $\m M(d+1)_{\bullet d+1}\in\bigtimes_{i=1}^{d+1}(\varepsilon,m_i^*-\varepsilon)$ with $\v m^*=(\diag[\m M(d)]\t,1)$ which implies that a solution is contained in a sufficiently large open ball $B_{r_{\varepsilon}}^{d+1}$. We apply Lemma \ref{lem:bijective} to complete the inductive step and the proof.
\end{proof}

%

% The logistic conditionals family

%

\section{The logistic conditionals family}
\label{sec:log cond family}
We denote by $\q_{\m A}^{\LogCo}$ the logistic conditionals family, that is the $\mu$-conditionals family with logistic link function $\logistic(x)\eqdef1/[1+\exp(-x)]$. This parametric family has been proposed by \cite{farrell2006nonlinear}, and in more general terms suggested by \cite{arnold1996distributions}. In this section, we motivate why the logistic link function arises somewhat naturally in the context of $\mu$-conditional families.
\begin{definition}
Let $\m A$ be a $d\times d$ real-valued lower triangular matrix. We refer to
\begin{equation*}
\q_{\m A}^{\ExpQu}(\v\gamma)=\exp(h+\v\gamma\t\m A\v\gamma),
\end{equation*}
as the exponential quadratic family with $h\eqdef-\log[\sum_{\v x\in\B^d}\exp(\v x\t\m A\v x)]$.
\end{definition}
\begin{proposition}
If $\m A=\diag(\v a)$, then $a_{ii}=\logistic^{-1}(m_{ii})$ and $\q_{\m A}^{\ExpQu}=\q_{\m A}^{\LogCo}=\q_{\v m}^{\Prod}$.
\end{proposition}
The exponential quadratic family is a natural way to design a parametric family and plays a central role in physics and life science being the well-studied Ising model on a weighted complete graph. It links to information theory \citep{soofi1994capturing}, log-linear theory for contingency tables \citep[ch. 5]{bishop75discrete} and graphical models \citep[ch. 2]{cox1996multivariate}. Finding its mode is an NP-hard problem and intensively studied in the field of operation research \citep[for a recent review]{boros2007local}.

Proposition \ref{prop:maxent} states that the exponential quadratic family is the maximum entropy distribution on $\B^d$ having a given cross-moment matrix. It appears to be the binary analogue of the multivariate normal distribution which is the maximum entropy distribution on $\R^{d}$ having a given covariance matrix \citep[sec. 5.1.1]{kapur1989maximum}. We can read the parameters $a_{ij}$ as Lagrange multipliers or, if $i\neq j$, as conditional log odd-ratios since
\begin{equation*}
a_{ij}=\log\left[\frac
{\probplus{\q_{\m A}^{\ExpQu}}{\Gamma_i=1,\Gamma_j=1 \mid \v \Gamma_{-i,j}}\probplus{\q_{\m A}^{\ExpQu}}{\Gamma_i=0,\Gamma_j=0 \mid \v \Gamma_{-i,j}}}
{\probplus{\q_{\m A}^{\ExpQu}}{\Gamma_i=0,\Gamma_j=1 \mid \v \Gamma_{-i,j}}\probplus{\q_{\m A}^{\ExpQu}}{\Gamma_i=1,\Gamma_j=0 \mid \v \Gamma_{-i,j}}}\right].
\end{equation*}
We might interpret the constant conditional log odd-ratios as analogue of the constant conditional correlations of the multivariate normal distribution \citep{wermuth1976analogies}.

Despite these similarities to the multivariate normal distribution, we cannot easily sample from the exponential quadratic family nor explicitly relate the parameter $\m A$ to the cross-moment matrix $\m M$. The reason is that the lower dimensional marginal distributions are difficult to compute \citep[(iii)]{cox1972analysis}.

\begin{proposition}
The marginal distribution of the exponential quadratic family is
\begin{equation}
\textstyle
\q_{\m A}^{\ExpQu}(\v\gamma_{-d})=
\exp\left(h+\v\gamma_{-d}\t\m A_{-d}\v\gamma_{-d}+\log\left[1+\exp(a_{dd}+\sum_{j=1}^{d-1}a_{ij}\gamma_j\right]\right).
\end{equation}
\end{proposition}
% \begin{proof}
% Straightforward, since
% \begin{align*}
% q_{\m A}^{\ExpQu}(\v\gamma_{-d})
% &=q_{\m A}^{\ExpQu}(\v\gamma_{d}=0,\v\gamma_{-d})+q_{\m A}^{\ExpQu}(\v\gamma_{d}=1,\v\gamma_{-d}) \\
% &=\textstyle\exp(h+\v\gamma_{-d}\t\m A_{-d}\v\gamma_{-d})
% \left\lbrack1+\exp(a_{d}+\sum_{j=1}^{d-1}a_{ij}\gamma_j)\right\rbrack.
% \end{align*}
% \end{proof}
We cannot repeat the marginalization since the multi-linear structure is lost. In fact, the following result shows that the logistic conditionals family is precisely constructed such that the non-linear term in the above expression vanishes.
\begin{proposition}
Let $\m A$ be a $d\times d$ lower triangular matrix. The logistic conditionals family can be written as
\begin{align*}
\q_{\m A}^{\LogCo}(\v\gamma)
&=\textstyle
\exp\left(\v\gamma\t\m A\v\gamma-\sum_{i=1}^d \log\left[1+\exp(a_{ii}+\sum_{j=1}^{i-1} a_{ij}\gamma_j)\right]\right).
\end{align*}
\end{proposition}
\begin{proof}
Straightforward calculations yield
\begin{align*}
\log \q_{\m A}^{\LogCo}(\v\gamma)
&\textstyle=
\sum_{i=1}^d \log\left([\logistic(a_{ii}+\sum_{j=1}^{i-1}a_{ij}\gamma_j)]^{\gamma_i}[1-\logistic(a_{ii}+\sum_{j=1}^{i-1}a_{ij}\gamma_j)]^{1-\gamma_i}\right) \\
&\textstyle=
\sum_{i=1}^d\left(\gamma_i\log[\logistic(a_{ii}+\sum_{j=1}^{i-1}a_{ij}\gamma_j)]+
(1-\gamma_i)\log[1-\logistic(a_{ii}+\sum_{j=1}^{i-1}a_{ij}\gamma_j)]\right) \\
&\textstyle=
\sum_{i=1}^d\left(\gamma_i\,\logistic^{-1}[\logistic(a_{ii}+\sum_{j=1}^{i-1}a_{ij}\gamma_j)]+
\log[1-\logistic(a_{ii}+\sum_{j=1}^{i-1}a_{ij}\gamma_j)]\right) \\
&\textstyle=
\sum_{i=1}^d\left(\gamma_i(a_{ii}+\sum_{j=1}^{i-1}a_{ij}\gamma_j)-
\log[1+\exp(a_{ii}+\sum_{j=1}^{i-1}a_{ij}\gamma_j)]\right) \\
&\textstyle=
\sum_{i=1}^d\sum_{j=1}^i a_{ij}\gamma_i\gamma_j-\sum_{i=1}^d \log[1+\exp(a_{ii}+\sum_{j=1}^{i-1} a_{ij}\gamma_j)] \\
&\textstyle=
\v\gamma\t\m A\v\gamma-\sum_{i=1}^d \log[1+\exp(a_{ii}+\sum_{j=1}^{i-1} a_{ij}\gamma_j)],
\end{align*}
where we used $\log[1-\logistic(x)]=-\log[1+\exp(x)]$ in the third line.
\end{proof}
The full conditional probability of the $d$-dimensional exponential quadratic family is a logistic regression term.
\begin{proposition}
The conditional distribution of the exponential quadratic family is
\begin{equation*}
\textstyle
\q_{\m A}^{\ExpQu}(\v\gamma_i=1\mid\v\gamma_{-i})=\logistic(a_{ii}+\sum_{j=1}^{i-1}a_{ij}\gamma_j+\sum_{j=i+1}^{d}a_{ji}\gamma_j).
\end{equation*}
\end{proposition}
% \begin{proof}
% Straightforward, since
% \begin{align*}
% q_{\m A}^{\ExpQu}(\v\gamma_i=1\mid\v\gamma_{-i})
% &=\frac{\exp(h+\v\gamma_{-i}\t\m A_{-i}\v\gamma_{-i}+a_{ii}+\sum_{j=1}^{i-1}a_{ij}\gamma_j+\sum_{j=i+1}^da_{ji}\gamma_j)}
%        {\exp(h+\v\gamma_{-i}\t\m A_{-i}\v\gamma_{-i})
%             \left\lbrack1+\exp(a_{ii}+\sum_{j=1}^{i-1}a_{ij}\gamma_j+\sum_{j=i+1}^da_{ji}\gamma_j)\right\rbrack}
% \end{align*}
% \end{proof}

Since we cannot repeat the marginalization for lower dimensions, we cannot assess the lower dimensional conditional probabilities which are necessary for sampling. We can, however, derive a series of approximate marginal probabilities that produce a logistic conditionals family which is, for low correlations, close to the original exponential quadratic family. This idea goes back to \citet{cox1994note}.

\begin{proposition}
Let $c_1+c_2x+c_3x^{2}\approx\log[\cosh(x)]$ be a second order approximation. We may approximate the marginal distribution $\q_{\m A}^{\ExpQu}(\v\gamma_{-d})$ by an exponential quadratic family $\exp(h_{*}+\v\gamma_{-d}\t\m A_{*}\v\gamma_{-d})$ with parameters
\begin{equation*}
h_*\eqdef\textstyle h+\log(2)+c_1+\frac{1}{2}a_{dd},\quad \m A_*\eqdef\textstyle\m A_{-d}+(c_2+\frac{1}{2})\diag(\v a_{*})+c_3\,\v a_{*}\v a_{*}\t,
\end{equation*}
where $\v a_{*}\eqdef(a_{d1},\dots,a_{d\,d-1})\t$ denotes the $d$th column of $\m A$ without $a_{dd}$.
\end{proposition}

\begin{proof}
We write the marginal distribution of the exponential quadratic family as
\begin{align*}
\q_{\m A}^{\ExpQu}(\v\gamma_{-d})
=\textstyle\exp\Big[h+
\v\gamma_{-d}\t\m A_{-d}\v\gamma_{-d}+\frac{1}{2}(a_{dd}+\v a_*\t\gamma_{-d})
\textstyle+\log\left(2\cosh\left[\frac{1}{2}(a_{dd}+\v a_*\t\gamma_{-d})\right]\right)\Big].
\end{align*}
using the identity
\begin{equation*}
\log[1+\exp(x)]
=\textstyle\log\big(\exp(\frac{1}{2}x)\big[\exp(-\frac{1}{2}x)+\exp(\frac{1}{2}x)\big]\big)
=\textstyle\frac{1}{2}x+\log\big[2\cosh(\frac{1}{2}x)\big]
\end{equation*}
and approximate the non-quadratic term by the second order polynomial
\begin{equation*}
\textstyle
\log[\cosh(\frac{1}{2}a_{dd}+\frac{1}{2}\v a_{*}\t\v\gamma_{-d})]\approx
c_1+c_2 \v a_{*}\t\v\gamma_{-d}+c_3 (\v a_{*}\t\v\gamma_{-d})^2.
\end{equation*}
We rewrite the inner products $\v a_{*}\t\v\gamma_{-d}+(\v a_{*}\v\gamma_{-d})^2 =\v\gamma_{-d}\t\left[\diag(\v a_{*})+\v a_{*}\v a_{*}\t\right]\v\gamma_{-d}$ and rearrange the quadratic terms.
\end{proof}

We can iterate the procedure to construct a logistic conditionals family which is close to the original exponential quadratic family. However, the function $\log[\cosh(x)]$ behaves like a quadratic function around zero and like the absolute value function for large $\vert x\vert$. Thus, a quadratic polynomial can only approximate $\log[\cosh(x)]$ well for small values of $x$ which means that exponential quadratic families with strong dependencies is hard to approximate. \citet{cox1994note} propose a Taylor approximation which fits well around $\frac{1}{2}a_{dd}$ and works for weak correlations. The parameters are $\textstyle\v c=\big(\log[\cosh(\frac{1}{2}a_{dd})]),\frac{1}{2}\tanh(\frac{1}{2}a_{dd}),\frac{1}{8}\,\mathrm{sech}^2(\frac{1}{2}a_{dd})\big)$.

% Alternatively, we define sampling points $x_1,\dots,x_n$, compute $y_k=\log\cosh(\frac{1}{2}a_{dd}+x_k)$ and use the least squares estimate
% \begin{align*}
% \v c=[(\v 1, \v x, \v x^2)\t(\v 1, \v x, \v x^2)]^{-1}(\v 1, \v x, \v x^2) \v y.
% \end{align*}
% This provides a better overall approximation, but the fit might be poor around $\frac{1}{2}a_{dd}$.

%

% Sampling binary data with specified cross-moment matrix

%

\section{Sampling binary data with specified cross-moment matrix}
\label{sec:sampling}
If $2^d-1$ full probabilities are known, we easily sample from the corresponding multinomial distribution \citep{walker1977efficient}. For a valid set of cross-moments $m_I$, $I\in\mathcal I$, \citet{gange1995generating} proposes to compute the full probabilities using a variant of the Iterative Proportional Fitting algorithm \citep{haberman1972algorithm}. While there are no restrictions on the range of dependencies, we have to enumerate the entire state space which limits this versatile approach to low dimensions.

In the sequel, we do not consider methods for structured correlations nor approaches which require enumeration of the state space. First, we show how to compute the parameter $\m A$ of a $\mu$-conditionals model for a given cross-moment matrix $\m M$. Secondly, we review an alternative approach to sampling binary data based on the multivariate normal distribution \citep{emrich1991method}.

\subsection{Fitting the conditionals family}
The proof of Theorem \ref{thm:mean+corr} suggests an iterative procedure to adjust the parameter $\m A$ to a given cross-moment matrix $\m M$. We add new cross-moments $\v m\in(0,1)^{d+1}$ to the $d\times d$ a lower triangular matrix $\m A$ by solving the non-linear equation $f(\v a)=\v m$ via Newton-Raphson iterations $\v a^{(k+1)}=\v a^{(k)}-[f'(\v a^{(k)})]^{-1}[f(\v a^{(k)})-\v m]$ where
\begin{align*}
f(\v a) &=\textstyle\sum_{\v\gamma\in\B^{d}} \q_{\m A}^{\Link}(\v\gamma)\link [(\v\gamma\t,1)\v a](\v\gamma\t,1)\t \\
f'(\v a)&=\textstyle\sum_{\v\gamma\in\B^{d}} \q_{\m A}^{\Link}(\v\gamma)\link'[(\v\gamma\t,1)\v a](\v\gamma\t,1)\t(\v\gamma\t,1)
\end{align*}
For dimensions $d>10$, the exact computation of the expectations becomes expensive, and we replace $f$ and $f'$ by their Monte Carlo estimates
\begin{equation}
\label{eq:MC estimates}
\begin{aligned}
\hat f(\v a) &=\textstyle\sum_{k=1}^n \q_{\m A}^{\Link}(\v\gamma)\link [(\v x_k\t,1)\v a)](\v x_k\t,1) \\
\hat f'(\v a)&=\textstyle\sum_{k=1}^n \q_{\m A}^{\Link}(\v\gamma)\link'[(\v x_k\t,1)\v a)](\v x_k\t,1)\t(\v x_k\t,1)
\end{aligned}
\end{equation}
where $\v x_1,\dots,\v x_n$ are drawn from $\q^{\Link}_{\m A}$. Some remarks are in order.
\begin{itemize}
\item If the smallest eigenvalue of $\m M-\diag(\m M)\diag(\m M)\t$ approaches zero or a cross-moment $m_{ij}$ approaches the bounds \eqref{eq:bin bounds}, the parameter $a_{ij}$ may become very large in absolute value. The limited numerical accuracy available on a computer inhibits sampling from such extreme cases.
\item We might encounter numerical trouble in the course of the fitting procedure. In order to circumvent problems, we set
\begin{equation*}
m_{ij}(\lambda_{k})\eqdef\lambda_{k} m_{ij}+(1-\lambda_{k})m_{ii}m_{jj},\quad 0=\lambda_1<\dots<\lambda_n=1
\end{equation*}
for all $j=1,\dots,i-1$ and compute a sequence of solutions $\v a(\lambda_k)$ to the cross-moments $\v m(\lambda_k)$. We stop if the parameters fail to converge which ensures that the mean of the $\mu$-conditionals family is always $\diag(\m M)$.
% \item Yet another way to tweak the numerical properties is reparameterisation through swapping the component $i$ and another component $j\in\set{i+1,\dots,d}$. Later, we have to apply the inverse permutation in the sampling algorithm to deliver the binary vector in the original order.
\item If we have data available instead of cross-moments, we would rather fit the family via component-wise likelihood maximization which is usually faster than the method of moments and can even be parallelized \citep{schaefer2011sequential}.
\item For the linear link function $\mu(x)=x$, we obtain
\begin{align*}
f(\v a)
=\textstyle\left[\sum_{\v\gamma\in\B^{d}} \q_{\m A}^{\Link}(\v\gamma)(\v\gamma\t,1)\t(\v\gamma\t,1)\right]\v a
=\begin{pmatrix}
\m M   & \v m \\
\v m\t & 1
\end{pmatrix}\v a
\end{align*}
which always has a solution by virtue of Lemma \ref{lem:ext cross-moments}; to construct a mass function, however, we have to fall back to the truncated version $\mu(x)=\min\set{\max\set{x,0},1}$, and the range of feasible cross-moments is hard to assess \citep{qaqish2003family}.
\end{itemize}

\subsection{Fitting the Gaussian copula family}
\label{sec:Gaussian}
\cite{emrich1991method} propose to dichotomize a multivariate Gaussian distribution for sampling multivariate binary data.
\begin{definition}
For a vector $\v a\in \R^d$ and a $d\times d$ correlation matrix $\g \Sigma$ we define the Gaussian copula family
\begin{equation*}
\label{eq:Gaussian copula}
\textstyle
\q_{\v a,\g \Sigma}^{\Gau}(\v\gamma)=\int_{\tau_{\v a}^{-1}(\v\gamma)} \varphi_{\g \Sigma}(\v x)\,d\v x,\quad
\varphi_{\g \Sigma}(\v x)=(2\pi)^{-d/2}\abs{\g \Sigma}^{-1/2}\exp\left(-\textstyle\frac{1}{2}\,\v x\t\g \Sigma^{-1}\v x\right),
\end{equation*}
where $\tau_{\v a}(\v x)\eqdef\left(\ind_{(-\infty,a_1]}(x_1),\dots,\ind_{(-\infty,a_d]}(x_d)\right)$.
\end{definition}
For all $I\subseteq D$, the marginals are
\begin{align*}
m_I
&=\textstyle
\sum_{\v\gamma\in\B^d} \q_{\v a,\g \Sigma}^{\Gau}(\v\gamma)\prod_{i\in I}\gamma_i
 =\sum_{\v\gamma\in\B^d,\,\v\gamma_I=\v 1} \int_{\tau_{\v a}^{-1}(\v\gamma)} \varphi_{\g \Sigma}(\v v)\,d\v v \\
&=\textstyle
\int_{\hspace{-2ex}\bigcup\limits_{\v\gamma\in\B^d,\,\v\gamma_I=\v 1}\hspace{-2ex}\set{\tau_{\v a}^{-1}(\v\gamma)}}\varphi_{\g \Sigma}(\v v)\,d\v v
=\int_{\times_{i=1}^{d}
\scriptscriptstyle
\begin{cases}\\[-1.9em]
\scriptscriptstyle (-\infty,a_i]    & \scriptstyle \hspace{-.7em} i\in I \\[-.5em]
\scriptscriptstyle (-\infty,\infty) & \scriptstyle \hspace{-.7em} i\notin I
\end{cases}}\, \varphi_{\g \Sigma}(\v v)\,d\v v=\Phi_{\g \Sigma}^{(I)}(\v a_I),
\end{align*}
where $\Phi_{\g \Sigma}^{(I)}$ is the marginal cumulative distribution function of the multivariate Gaussian. We set $a_i=\Phi^{-1}(m_i)$ for $i\in D$ to adjust the mean. In order to compute the parameter $\g \Sigma$ that yields the desired cross-moments $\m M$, we may use a fast series approximations \citep{drezner_98} to solve $m_{ij}=\Phi_{\sigma_{ij}}(a_i,a_j)$ for $\sigma_{ij}$ via Newton-Raphson iterations $\sigma_{ij}^{r+1}=\sigma_{ij}^{r}-[\Phi_{\sigma_{ij}^r}(a_i,a_j)-m_{ij}]/\varphi_{\sigma_{ij}^r}(a_i,a_j)$; \cite{modarres2011high} suggests the bivariate \cite{plackett1965class} distribution as a proxy for $\varphi_{\sigma_{ij}}$ which might provide a good starting value $\sigma_{ij}^{0}\in(-1,1)$.

While we always obtain a solution in the bivariate case, it is well-known that the resulting matrix $\g \Sigma$ is not necessarily positive definite due to the range of the Gaussian copula which allows to attain the bounds \eqref{eq:bin bounds} for $d\leq2$, but not for higher dimensions. In that case, we can replace $\g \Sigma$ by 
\begin{equation}
\label{eq:lower Sigma}
\g \Sigma^*=(\g \Sigma+\abs{\lambda}\m I)/(1+\abs{\lambda})>0
\end{equation}
where $\lambda$ is smaller than any eigenvalue of $\g \Sigma$. Alternatively, we can project $\g \Sigma$ into the set of correlation matrices; see \citet{higham_02} and follow-up papers for algorithms that compute the nearest correlation matrix in Frobenius norm.

The point-wise evaluation of $\q_{\v a,\g \Sigma}^{\Gau}(\v\gamma)$ requires the computation of multivariate normal probabilities, that is high-dimensional integrals with the respect to the density of the multivariate normal distribution. This is a computationally challenging task in itself \citep[see e.g.][]{genz2009computation}, and the Gaussian copula family is therefore not easily incorporated into the Markov chain Monte Carlo algorithms briefly discussed in the introduction.

\section{Numerical experiments}
\label{sec:numerical experiments}
In this section, we compare the $\mu$-conditionals family with truncated linear and logistic link function to the Gaussian copula family. We draw random cross-moment matrices of varying dimension and difficulty, fit the parametric families and record how well the desired correlation structure can be reproduced on average.

\subsection{Random cross-moments}
\label{sec:random matrices}
We first sample the mean $\v m=\diag(\m M)\sim\uni_{(0,1)^d}$. For the off-diagonal elements, we have to ensure that the covariance matrix $\m M-\v m\v m\t$ is positive definite and that the constraints \eqref{eq:bin bounds} are all met. We alternate the following two steps.

\begin{itemize}
\item Permutations $m_{ij}=m_{\sigma(i)\sigma(j)}$ for $i,j\in D$ with uniform $\sigma\sim\uni_{S(D)}$ where we denote by $S(D)\eqdef\set{\sigma\colon D\to D,\sigma\text{ is bijective}}$ the set of all permutations on $D$.
\item Replacements $m_{id}=m_{di}\sim\uni_{[a_{i},b_{i}]}$ for all $i=\sigma(1),\dots,\sigma(d-1)$ with uniform $\sigma\sim\uni_{S(D\setminus\set d)}$ where the bounds $a_{i},b_{i}$ are subject to the constraints $\det(\m M)>0$ and $\min\set{m_{ii}+m_{dd}-1,0}\leq m_{id}\leq \max\set{m_{ii},m_{dd}}$.
\end{itemize}

The replacement step needs some consideration. We denote by $\m N$ the inverse of the $(d-1)\times(d-1)$ upper sub-matrix of $\m M$ and define $\tau_{i}\eqdef m_{di}\sum_{i\in D\setminus\set d}m_{dj}n_{ij}$ such that
$\textstyle\det(\m M)=[1/\det(\m N)] (m_{dd}-\sum_{i\in D\setminus\set d}\tau_i).$ If we replace $m_{di}=m_{id}$ by $x_i$ we have to ensure that $\det[\m M(x_i)]=\det(\m M)+m_{di}(m_{di}n_{ii}+2\tau_{i})-x_i(x_in_{ii}+2\tau_{i})>0$ which means $(x_i+\tau_{i}/n_{ii})\in(-c_{i},c_{i})$ with $c_{i}\eqdef[\tau_{i}^{2}/n^{2}_{ii}+\det(\m M)+m_{di}(m_{di}n_{ii}+2\tau_{i})]^{-1/2}$. Therefore, the lower and upper bounds, $a_{i}\eqdef\max\set{m_{ii}+m_{dd}-1,0,-\tau_{i}/n_{ii}-c_{i}}$ and $b_{i}\eqdef\min\set{m_{ii},m_{dd},-\tau_{i}/n_{ii}+c_{i}}$, respect all constraints on $x_i$. We rapidly update the value of the determinant $\det[\m M(x_i)]$ and proceed with the next entry.

We perform $10\cdot d$ permutation steps and run $500$ sweeps of replacements between permutations. The result is approximately a uniform draw from the set of feasible cross-moments matrices. However, sampling according to these cross-moments might not be possible in higher dimensions because the cross-moment matrix is likely to contain extreme cases which are beyond the scope of the parametric family or not workable for numerical reasons. We introduce a parameter $\varrho\in[0,1]$ which governs the difficulty of the sampling problem by shrinking the upper and lower bounds $a$ and $b$ of the uniform distributions to $a^\varrho\eqdef[(1+\varrho)a+(1-\varrho)b]/2$ and $b^\varrho\eqdef[(1-\varrho)a+(1+\varrho)b]/2$, respectively.

\subsection{Figure of merit}
Let $\m M$ be a cross-moments matrix and let $\m M^{*}$ denote the cross-moment matrix with mean $\v m=\diag(\m M)$ and uncorrelated entries $m^{*}_{ij}=m_{ii}m_{jj}$ for all $i\neq j\in D$. For a parametric family $\q_{\theta}$, we define the figure of merit
\begin{equation}
\label{eq:quantity}
\tau_{\q}(\m M)\eqdef(\norm{\m M-\m M^{*}}-\norm{\m M-\m M^\q})/\norm{\m M-\m M^{*}},
\end{equation}
where $\m M^{\q}$ denotes the sampling cross-moment matrix of the parametric family with parameter $\theta$ adjusted to the desired cross-moment matrix $\m M$. The norm $\norm{\cdot}$ might be any non-trivial matrix norm; in our numerical experiments we use the spectral norm $\norm{\m A}_{2}^{2}\eqdef\lambda_{\max}(\m A\t\m A)$, where $\lambda_{\max}$ delivers the largest eigenvalue, but we found the Frobenius norm $\norm{\m A}_{F}^{2}\eqdef\tr{\m A\m A\t}$ to provide qualitatively the very same picture.

% We can roughly interpret $\tau_{\q}(\m M)$ as the proportion of the correlation structure that the parametric family is able to reproduce. The score $\tau_{q}(\m M)$ is negative if the parametric family $\q_{\theta}$ performs worse than $\q_{\v m}^{\Prod}$.

\subsection{Computational results}
For fitting the logistic conditionals family when $d>10$, we replace the exact terms by Monte Carlo estimates \eqref{eq:MC estimates} where we use $n=10^4$ random samples. We estimate the cross-moment matrix of the parametric family $\q$ by $\m M^{\q}\approx n^{-1}\sum_{k=1}^{n}\v x_{k}\v x_{k}\t$ where we use $n=10^6$ samples from $\q$. This concerns only the logistic and linear conditionals families; for the Gaussian copula family, we can explicitly compute the sampling cross-moments as $m^{\q}_{ij}=\Phi_2(\link_i,\link_j;\sigma_{ij})$, where $\g \Sigma$ is the adjusted correlation matrix of the underlying multivariate normal distribution made feasible via \eqref{eq:lower Sigma}.

%We set the differences $m_{ij}-m_{ij}^{q}$ to zero in \eqref{eq:quantity} if $\vert m_{ij}-m_{ij}^{\q} \vert<10^{-2}$ in order to avoid the accumulation of noise introduced by the Monte Carlo simulations.
We loop over $15$ levels of difficulty $\varrho\in[0,1]$ in $3$ dimensions $d=10,25,50$, and generate at each time $200$ cross-moments matrices. We denote by $\tau_1\leq\cdots\leq\tau_{200}$ the ordered figures of merit of the random cross-moment matrices. We report the median and the quantiles $(\tau_{\lfloor(0.5-\omega)n\rfloor},\tau_{\lceil(0.5+\omega)n\rceil})$, depicted as underlying gray areas for $20$ equidistant values of $\omega\in[0.0,0.5]$. Figures 1-3 show the results grouped by parametric families; the $y$-axis with the scale on the left represents the figure of merit $\tau\in[0,1]$, the $x$-axis represents the level of difficulty $\varrho\in[0,1]$, and the $[0.0,0.5]$-gray-scale on the right refers to the level of the quantiles.\\[-2em]

\begin{figure}[H]
\begin{center}
\caption{Logistic conditionals family}
\small $d=10$\hspace{28mm}$d=25$\hspace{28mm}$d=50$\\
\includegraphics[height=25mm, trim=0 0 1cm 0, clip=true]{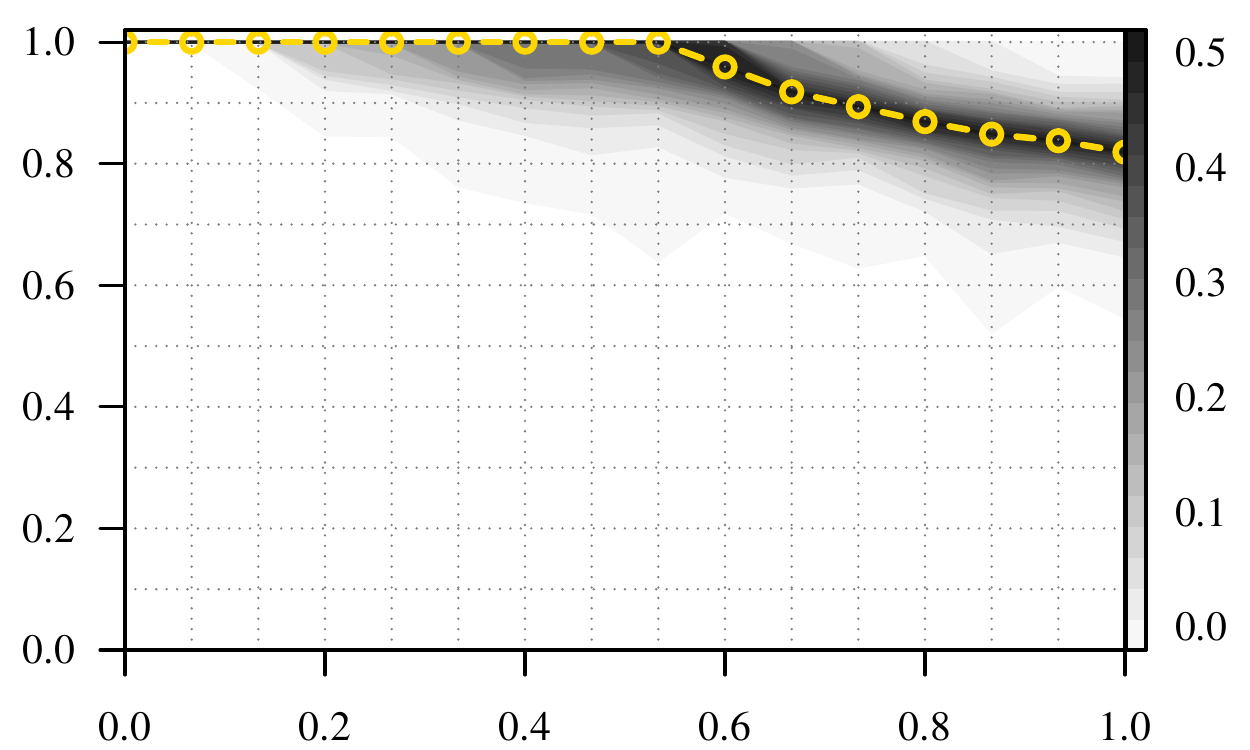}
\includegraphics[height=25mm, trim=1cm 0 1cm 0, clip=true]{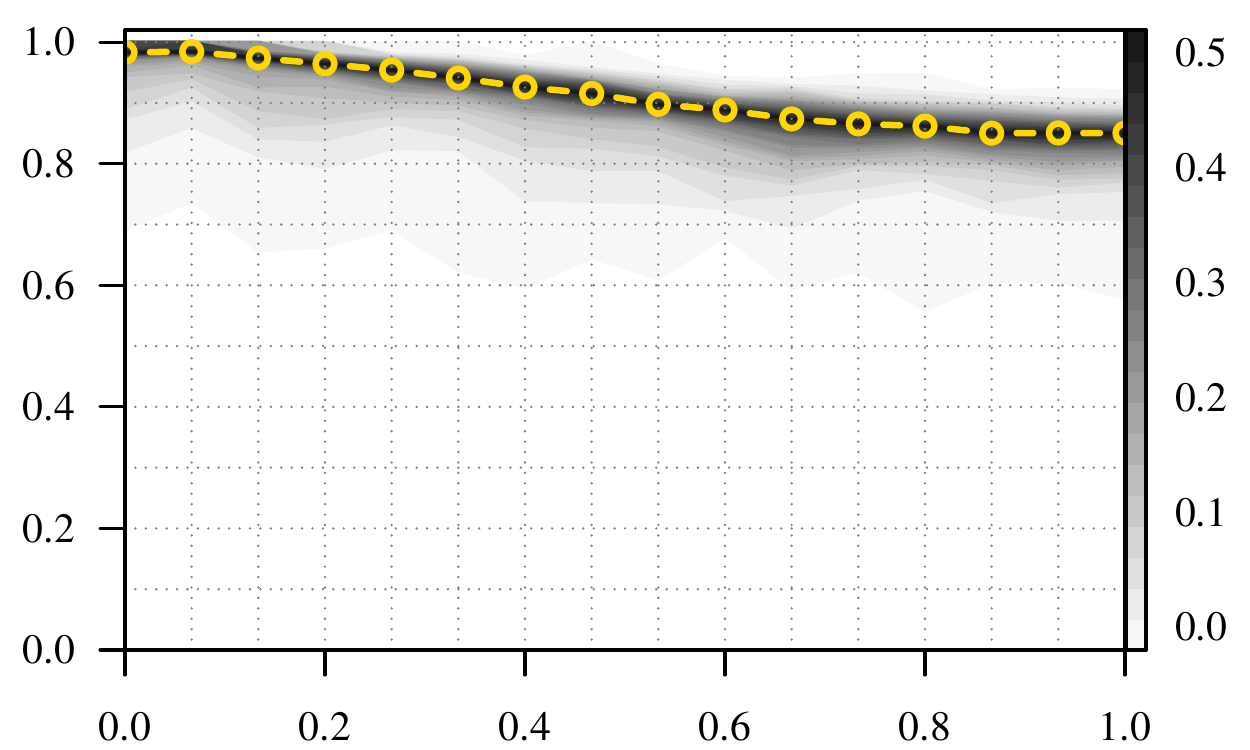}
\includegraphics[height=25mm, trim=1cm 0 0 0, clip=true]{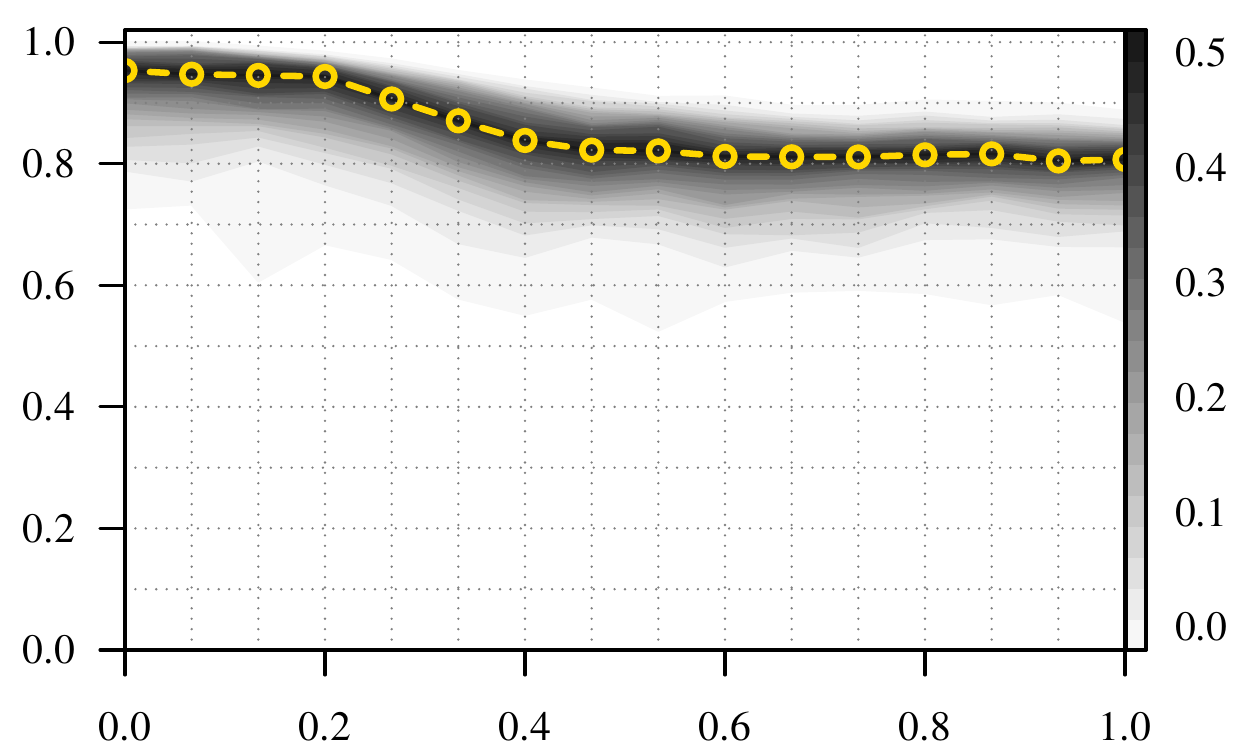}
\end{center}
\vspace{-3em}
\end{figure}
\begin{figure}[H]
\begin{center}
\caption{Gaussian copula family}
\small $d=10$\hspace{28mm}$d=25$\hspace{28mm}$d=50$\\
\includegraphics[height=25mm, trim=0 0 1cm 0, clip=true]{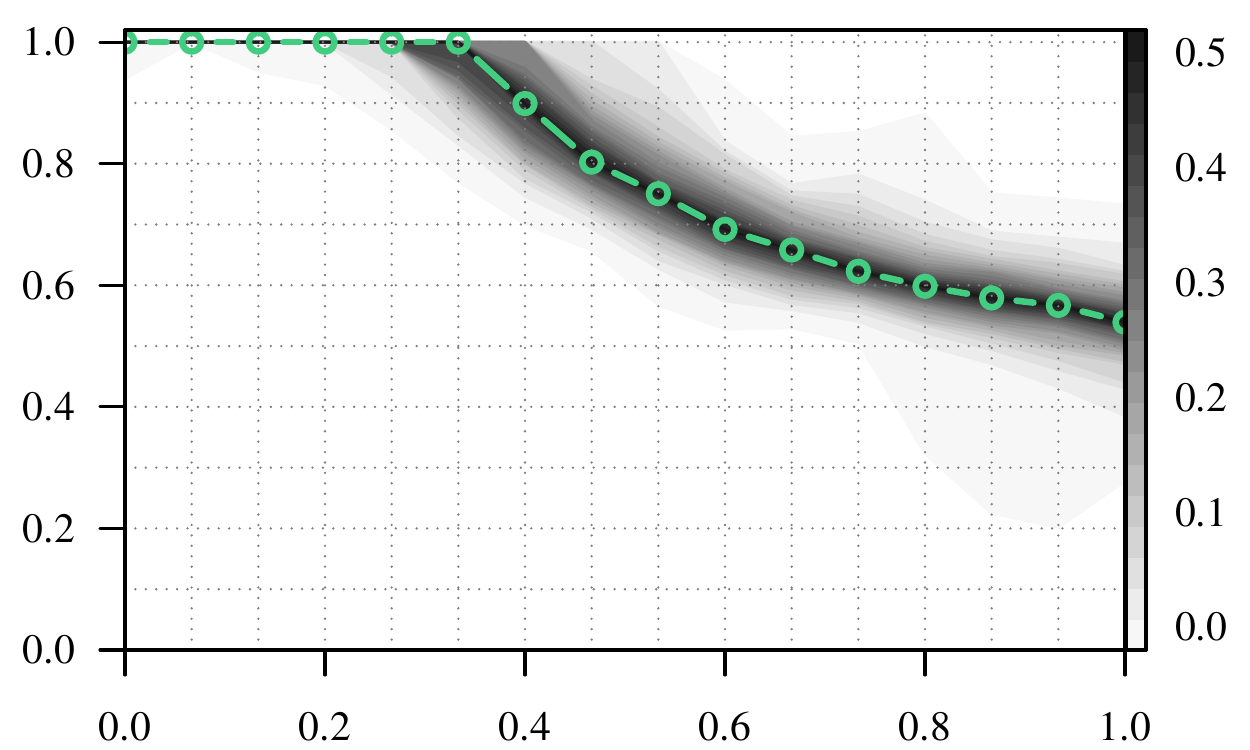}
\includegraphics[height=25mm, trim=1cm 0 1cm 0, clip=true]{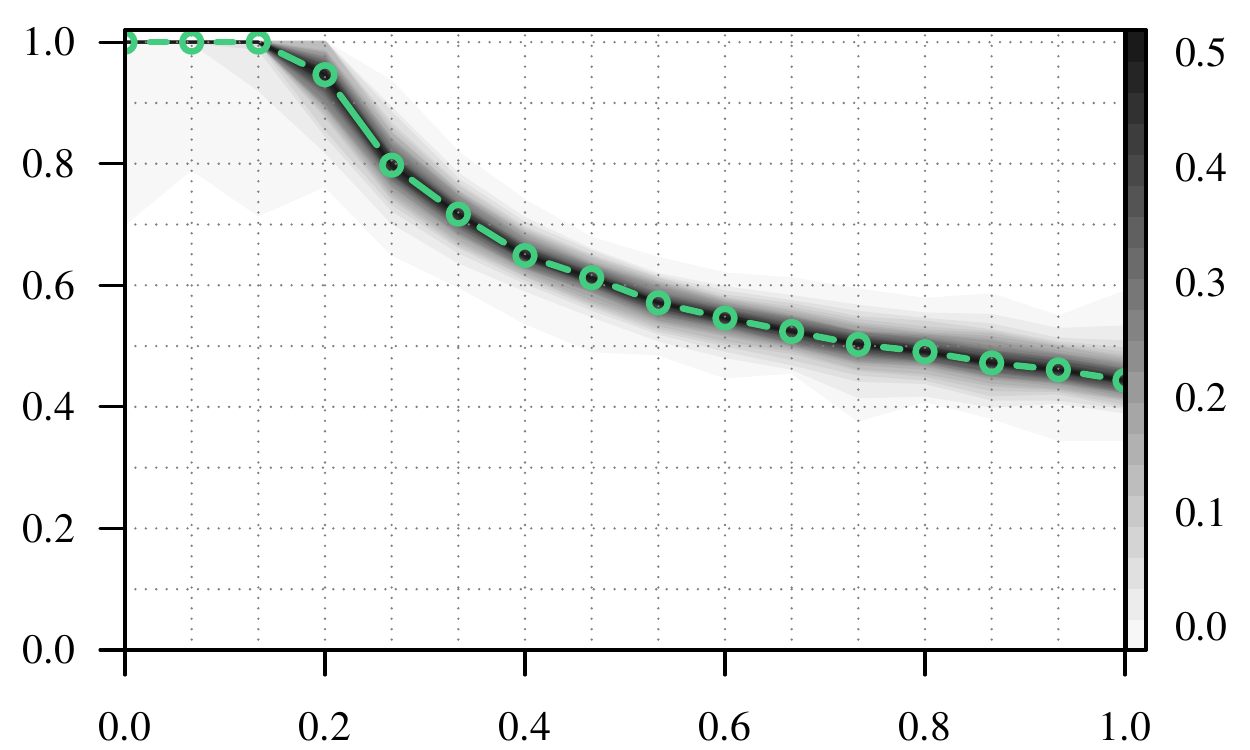}
\includegraphics[height=25mm, trim=1cm 0 0 0, clip=true]{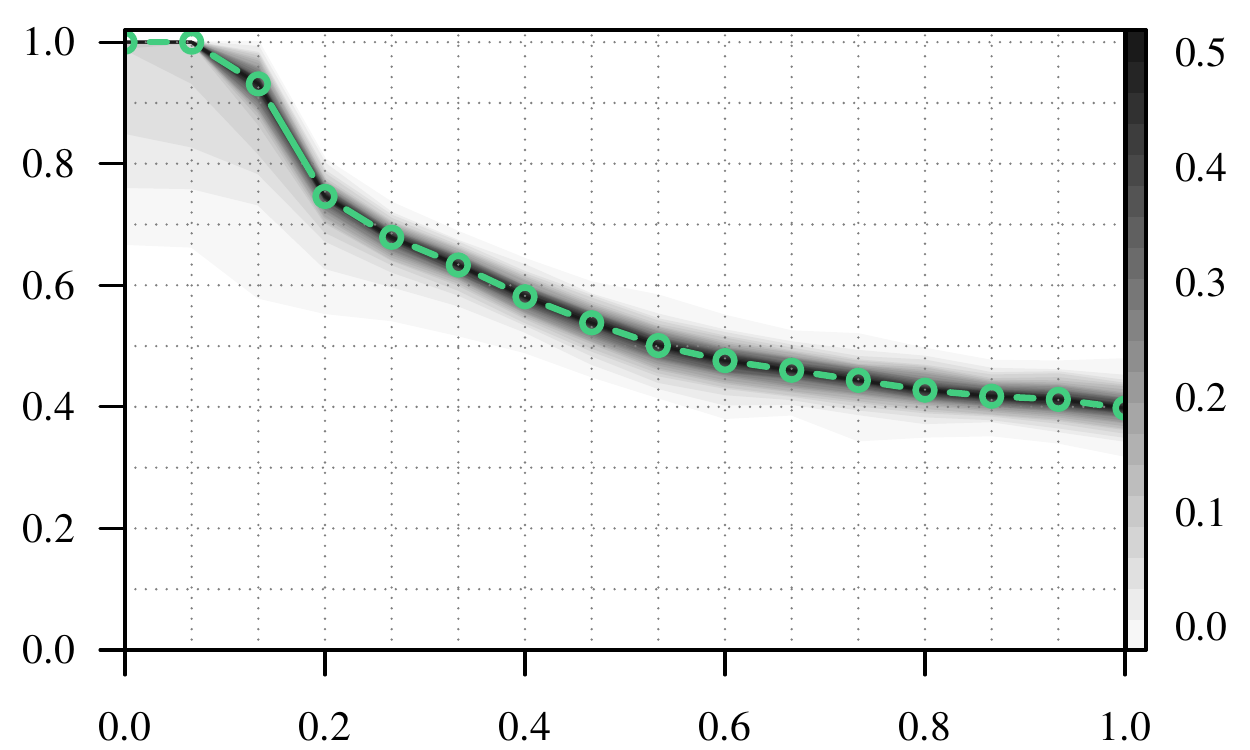}\\
\end{center}
\vspace{-3em}
\end{figure}
\begin{figure}[H]
\begin{center}
\caption{Truncated linear conditionals family}
\small $d=10$\hspace{28mm}$d=25$\hspace{28mm}$d=50$\\
\includegraphics[height=25mm, trim=0 0 1cm 0, clip=true]{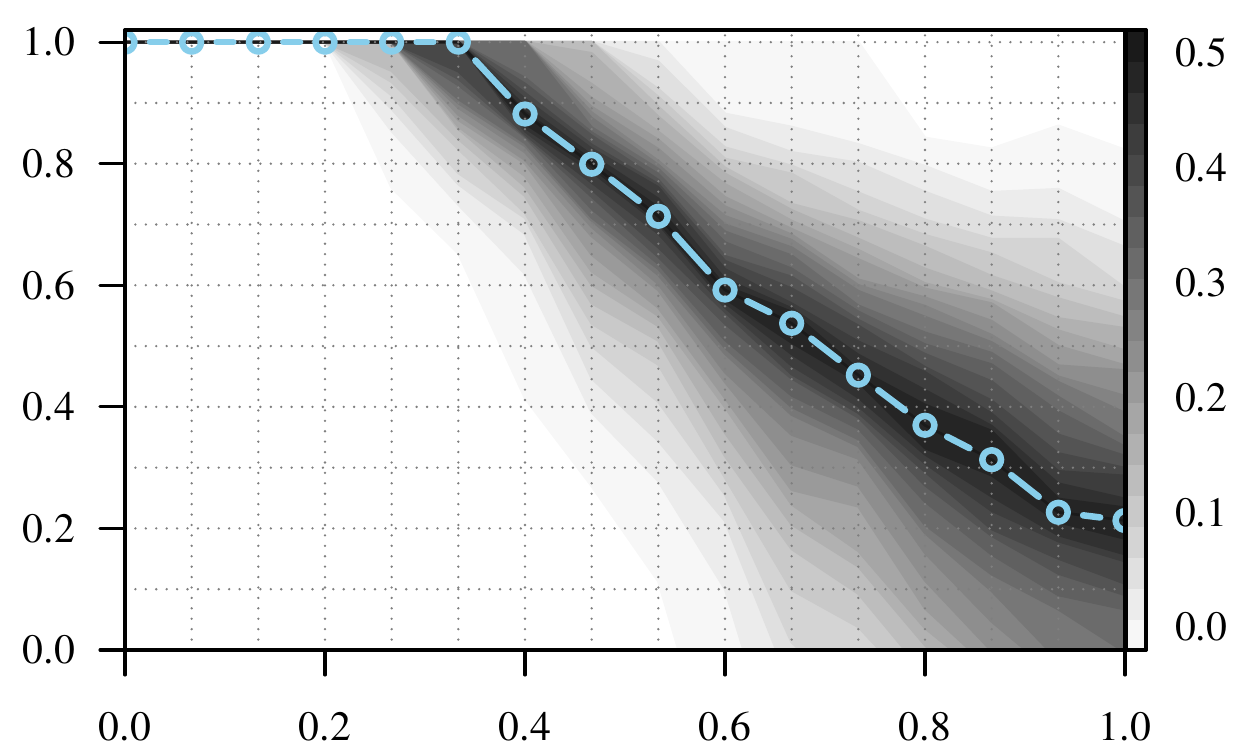}
\includegraphics[height=25mm, trim=1cm 0 1cm 0, clip=true]{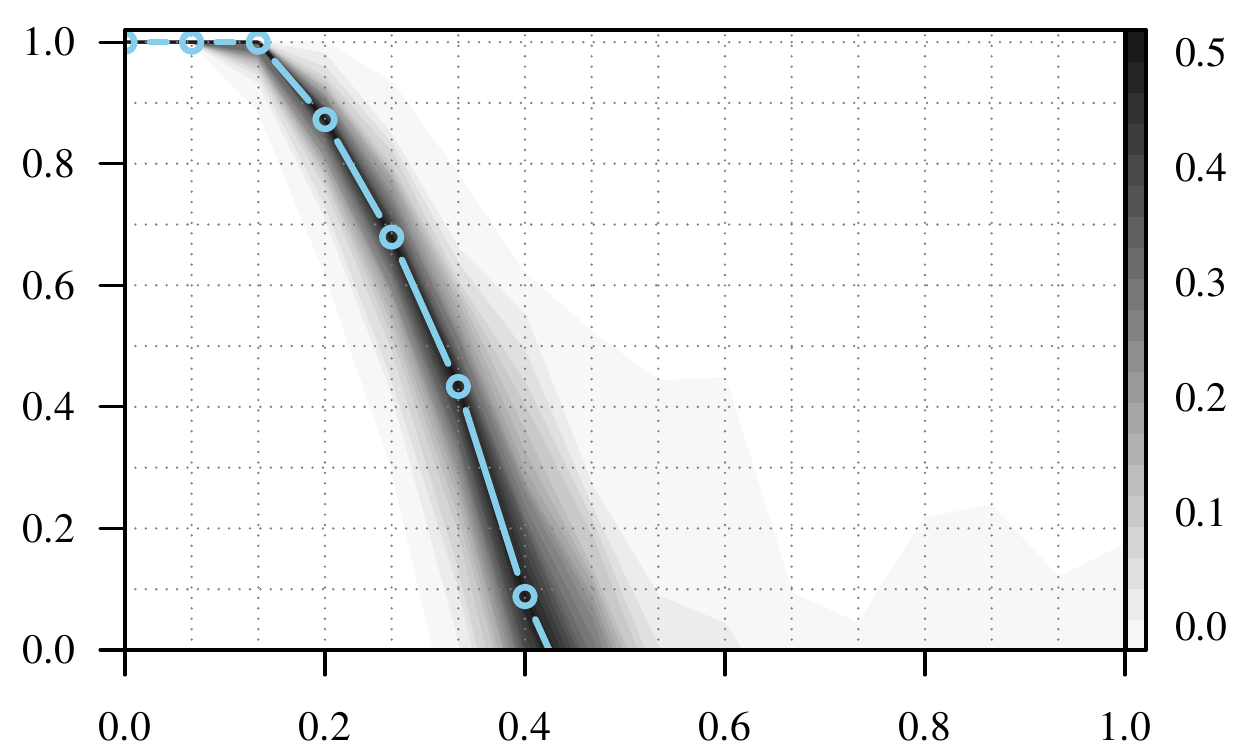}
\includegraphics[height=25mm, trim=1cm 0 0 0, clip=true]{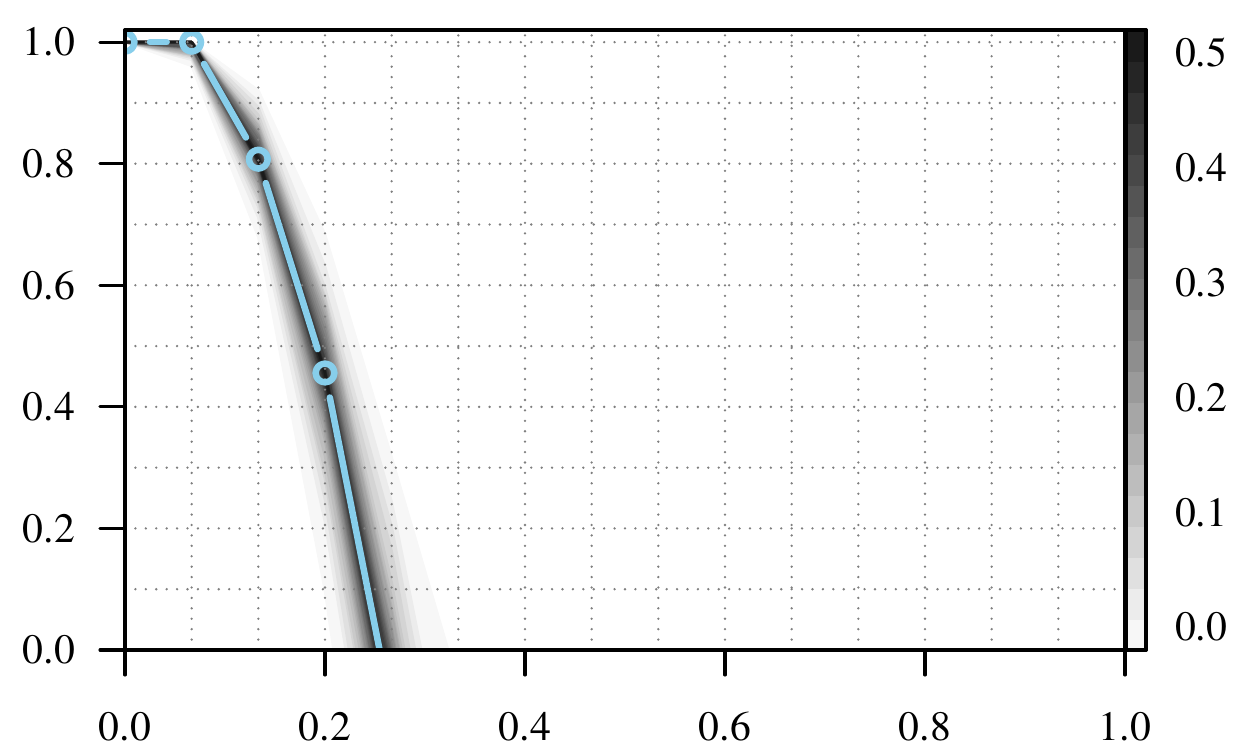}
\end{center}
\end{figure}

\subsection{Discussion}
While Theorem \ref{thm:mean+corr} suggests that the scope of the logistic conditionals family is far beyond competing approaches, we cannot, in practice, expect a binary parametric family with $d(d-1)/2$ dependency parameters to produce just any desired correlation structure. However, the practical scope of the logistic family is limited only by the available numerical accuracy while the scope of competing methods is also limited by their mathematical structure.

The truncated linear conditionals family is fast to compute but its quality deteriorates rapidly with growing complexity. The Gaussian copula family is guaranteed to have the correct mean but it is less flexible than the logistic conditionals family; besides, it does not allow for point-wise evaluation of its mass function. The logistic conditionals family is computationally demanding but by far the most versatile option. These findings confirm similar comparisons carried out against the backdrop of particular applications \citep{farrell2008methods,schaefer2012particle}.

\section{Acknowledgements}
This work is part of the author's Ph.D. thesis at CREST under supervision of Nicolas Chopin whom I would like to thank for numerous discussions on this topic. I thank Ioannis Kosmidis for his comments on a prior version of this paper.

% end content

%\bibliographystyle{apalike}
%\bibliography{../../../References/ref}

\begin{thebibliography}{}

\bibitem[Arnold, 1996]{arnold1996distributions}
Arnold, B. (1996).
\newblock {Distributions with logistic marginals and/or conditionals}.
\newblock {\em Lecture Notes-Monograph Series}, 28:15--32.

\bibitem[Bahadur, 1961]{bahadur61representation}
Bahadur, R. (1961).
\newblock {A representation of the joint distribution of responses to n
  dichotomous items}.
\newblock In Solomon, H., editor, {\em {Studies in Item Analysis and
  Prediction}}, pages pp. 158--68. Stanford University Press.

\bibitem[Bishop et~al., 1975]{bishop75discrete}
Bishop, Y., Fienberg, S., and Holland, P. (1975).
\newblock {\em {Discrete multivariate analysis: Theory and Practice}}.
\newblock Cambridge, MA: MIT Press.

\bibitem[Boros et~al., 2007]{boros2007local}
Boros, E., Hammer, P., and Tavares, G. (2007).
\newblock {Local search heuristics for quadratic unconstrained binary
  optimization (QUBO)}.
\newblock {\em Journal of Heuristics}, 13(2):99--132.

\bibitem[Bottolo and Richardson, 2010]{bottolo2010ess}
Bottolo, L. and Richardson, S. (2010).
\newblock {Evolutionary stochastic search for Bayesian model exploration}.
\newblock {\em Bayesian Analysis}, 5(3):583--618.

\bibitem[Chaganty and Joe, 2006]{chaganty2006range}
Chaganty, N. and Joe, H. (2006).
\newblock {Range of correlation matrices for dependent Bernoulli random
  variables}.
\newblock {\em Biometrika}, 93(1):197--206.

\bibitem[Cox, 1972]{cox1972analysis}
Cox, D. (1972).
\newblock {The analysis of multivariate binary data}.
\newblock {\em Applied Statistics}, pages 113--120.

\bibitem[Cox and Wermuth, 1994]{cox1994note}
Cox, D. and Wermuth, N. (1994).
\newblock {A note on the quadratic exponential binary distribution}.
\newblock {\em Biometrika}, 81(2):403--408.

\bibitem[Cox and Wermuth, 1996]{cox1996multivariate}
Cox, D. and Wermuth, N. (1996).
\newblock {\em {Multivariate dependencies: Models, analysis and
  interpretation}}, volume~67.
\newblock Chapman \& Hall/CRC.

\bibitem[Dolnicar and Leisch, 2001]{dolnicar2001behavioral}
Dolnicar, S. and Leisch, F. (2001).
\newblock Behavioral market segmentation of binary guest survey data with
  bagged clustering.
\newblock {\em Artificial Neural Networks---ICANN 2001}, 2130:111--118.

\bibitem[Drezner and Wesolowsky, 1990]{drezner_98}
Drezner, Z. and Wesolowsky, G.~O. (1990).
\newblock {On the computation of the bivariate normal integral}.
\newblock {\em Journal of Statistical Computation and Simulation}, 35:101--107.

\bibitem[Emrich and Piedmonte, 1991]{emrich1991method}
Emrich, L. and Piedmonte, M. (1991).
\newblock {A method for generating high-dimensional multivariate binary
  variates}.
\newblock {\em The American Statistician}, 45:302--304.

\bibitem[Farrell and Rogers-Stewart, 2008]{farrell2008methods}
Farrell, P. and Rogers-Stewart, K. (2008).
\newblock Methods for generating longitudinally correlated binary data.
\newblock {\em International Statistical Review}, 76(1):28--38.

\bibitem[Farrell and Sutradhar, 2006]{farrell2006nonlinear}
Farrell, P. and Sutradhar, B. (2006).
\newblock {A non-linear conditional probability model for generating correlated
  binary data}.
\newblock {\em Statistics \& probability letters}, 76(4):353--361.

\bibitem[Gange, 1995]{gange1995generating}
Gange, S. (1995).
\newblock {Generating Multivariate Categorical Variates Using the Iterative
  Proportional Fitting Algorithm}.
\newblock {\em The American Statistician}, 49(2).

\bibitem[Genz and Bretz, 2009]{genz2009computation}
Genz, A. and Bretz, F. (2009).
\newblock {\em {Computation of multivariate normal and t probabilities}},
  volume 195.
\newblock Springer.

\bibitem[George and McCulloch, 1997]{george_mcculloch_97}
George, E.~I. and McCulloch, R.~E. (1997).
\newblock {Approaches for Bayesian variable selection}.
\newblock {\em Statistica Sinica}, 7:339--373.

\bibitem[Haberman, 1972]{haberman1972algorithm}
Haberman, S. (1972).
\newblock {Algorithm AS 51: Log-linear fit for contingency tables}.
\newblock {\em Journal of the Royal Statistical Society. Series C (Applied
  Statistics)}, 21(2):218--225.

\bibitem[Hamze et~al., 2011]{hamze2011selfavoiding}
Hamze, F., Wang, Z., and de~Freitas, N. (2011).
\newblock {Self-Avoiding Random Dynamics on Integer Complex Systems}.
\newblock Technical report, arXiv:1111.5379.

\bibitem[Higham, 2002]{higham_02}
Higham, N.~J. (2002).
\newblock {Computing the nearest correlation matrix --- a problem from
  finance}.
\newblock {\em IMA Journal of Numerical Analysis}, 22:329--343.

\bibitem[Kapur, 1989]{kapur1989maximum}
Kapur, J. (1989).
\newblock {\em Maximum-entropy models in science and engineering}.
\newblock John Wiley \& Sons.

\bibitem[Lebbah et~al., 2008]{lebbah2008probabilistic}
Lebbah, M., Bennani, Y., and Rogovschi, N. (2008).
\newblock A probabilistic self-organizing map for binary data topographic
  clustering.
\newblock {\em International Journal of Computational Intelligence and
  Applications}, 7(4):363--383.

\bibitem[Lunn and Davies, 1998]{lunn1998note}
Lunn, A. and Davies, S. (1998).
\newblock {A note on generating correlated binary variables}.
\newblock {\em Biometrika}, 85(2):487--490.

\bibitem[McCullagh and Nelder, 1989]{mccullagh1989generalized}
McCullagh, P. and Nelder, J.~A. (1989).
\newblock {\em {Generalized Linear Models}}.
\newblock Chapman \& Hall / CRC, London.

\bibitem[Modarres, 2011]{modarres2011high}
Modarres, R. (2011).
\newblock High dimensional generation of bernoulli random vectors.
\newblock {\em Statistics \& Probability Letters}.

\bibitem[Nelsen, 2006]{nelsen2006introduction}
Nelsen, R. (2006).
\newblock {\em {An introduction to copulas}}.
\newblock Springer Verlag.

\bibitem[Oman and Zucker, 2001]{oman2001modelling}
Oman, S. and Zucker, D. (2001).
\newblock {Modelling and generating correlated binary variables}.
\newblock {\em Biometrika}, 88(1):287.

\bibitem[Park et~al., 1996]{park1996simple}
Park, C., Park, T., and Shin, D. (1996).
\newblock {A simple method for generating correlated binary variates}.
\newblock {\em The American Statistician}, 50(4).

\bibitem[Plackett, 1965]{plackett1965class}
Plackett, R. (1965).
\newblock A class of bivariate distributions.
\newblock {\em Journal of the American Statistical Association}, pages
  516--522.

\bibitem[Qaqish, 2003]{qaqish2003family}
Qaqish, B. (2003).
\newblock {A family of multivariate binary distributions for simulating
  correlated binary variables with specified marginal means and correlations}.
\newblock {\em Biometrika}, 90(2):455.

\bibitem[Robert and Casella, 2004]{RobCas}
Robert, C. and Casella, G. (2004).
\newblock {\em {Monte Carlo statistical methods}}.
\newblock Springer Verlag.

\bibitem[Rubinstein, 1999]{Rub:CE2}
Rubinstein, R.~Y. (1999).
\newblock {The Cross-Entropy Method for combinatorial and continuous
  optimization}.
\newblock {\em Methodology and Computing in Applied Probability}, 1:127--190.

\bibitem[Sch{\"a}fer, 2012]{schaefer2012particle}
Sch{\"a}fer, C. (2012).
\newblock {Particle algorithms for optimization on binary spaces}.
\newblock {\em pre-print}.
\newblock arXiv:1111.0574v1.

\bibitem[Sch{\"a}fer and Chopin, 2011]{schaefer2011sequential}
Sch{\"a}fer, C. and Chopin, N. (2011).
\newblock {Sequential Monte Carlo on large binary sampling spaces}.
\newblock {\em Statistics and Computing}, to appear.
\newblock doi: 10.1007/s11222--011--9299--z.

\bibitem[Soofi, 1994]{soofi1994capturing}
Soofi, E. (1994).
\newblock {Capturing the Intangible Concept of Information}.
\newblock {\em Journal of the American Statistical Association}, 89:1243--54.

\bibitem[Swendsen and Wang, 1987]{swendsen1987nonuniversal}
Swendsen, R. and Wang, J. (1987).
\newblock Nonuniversal critical dynamics in monte carlo simulations.
\newblock {\em Physical Review Letters}, 58(2):86.

\bibitem[Walker, 1977]{walker1977efficient}
Walker, A. (1977).
\newblock {An efficient method for generating discrete random variables with
  general distributions}.
\newblock {\em ACM Transactions on Mathematical Software}, 3(3):256.

\bibitem[Wermuth, 1976]{wermuth1976analogies}
Wermuth, N. (1976).
\newblock Analogies between multiplicative models in contingency tables and
  covariance selection.
\newblock {\em Biometrics}, pages 95--108.

\end{thebibliography}

\end{document}